\documentclass[11pt,a4paper]{article}

\usepackage{amsfonts,amsmath,amssymb,amsthm}
\usepackage{latexsym,amscd}
\usepackage{amsbsy}
\usepackage{CJK}
\usepackage{indentfirst,mathrsfs}
\usepackage{latexsym,amscd}
\usepackage{amsbsy}
\usepackage[noadjust]{cite}
\usepackage{color}
\usepackage{multirow}
\usepackage{makecell}
\usepackage{float}
\usepackage{booktabs}
\usepackage{multirow}
\usepackage{authblk}

\newtheorem{theorem}{Theorem}
\newtheorem{lem}{Lemma}
\newtheorem{cor}{Corollary}

\newtheorem{example}{Example}
\newtheorem{remark}{Remark}

\newcommand{\fq}{{\mathbb F}_{q}}
\newcommand{\fp}{{\mathbb F}_{p}}
\newcommand{\fqm}{{\mathbb F}_{q^m}}
\newcommand{\Tr}{{\rm {Tr}}}

\newcommand{\C}{{\mathcal{C}}}

\begin{document}

\title{Several classes of linear codes with few weights derived from \\Weil sums}

\author{Zhao Hu
\thanks{Z. Hu and N. Li are with the Key Laboratory of Intelligent Sensing System and Security (Hubei University), Ministry of Education, the Hubei Provincial Engineering Research Center of Intelligent Connected Vehicle Network Security, School of Cyber Science and Technology, Hubei University, Wuhan 430062. N. Li is also with the State Key Laboratory of Integrated Service Networks, Xi'an 710071, China. Email: zhao.hu@aliyun.com, nian.li@hubu.edu.cn},
Mingxiu Qiu
\thanks{M. Qiu is with the Hubei Key Laboratory of Applied Mathematics, Faculty of Mathematics and Statistics, Hubei University, Wuhan, 430062, China. Email: mingxiu.qiu@aliyun.com},
Nian Li, Xiaohu Tang
\thanks{X. Tang is with the Information Coding $\&$ Transmission Key Lab of Sichuan Province, CSNMT Int. Coop. Res. Centre (MoST), Southwest Jiaotong University, Chengdu, 610031, China. Email: xhutang@swjtu.edu.cn}
, Liwei Wu
\thanks{
L. Wu is with the Wuhan Maritime Communication Research Institute.
Email: 2777642@qq.com}
}
\date{}
\maketitle
\vspace{-4mm}
\begin{abstract}
  Linear codes with few weights have applications in secret sharing, authentication codes, association schemes and strongly regular graphs. In this paper, several classes of $t$-weight linear codes over $\fq$ are presented with the defining sets given by the intersection, difference and union of two certain sets, where $t=3,4,5,6$ and $q$ is an odd prime power. By using Weil sums and Gauss sums, the parameters and weight distributions of these codes are determined completely. Moreover, three classes of optimal codes meeting the Griesmer bound are obtained, and computer experiments show that many (almost) optimal codes can be derived from our constructions.
\end{abstract}
\noindent\textbf{Keywords:} Linear code $\cdot$ Weight distribution $\cdot$ Weil sum $\cdot$ Gauss sum.\\
\noindent\textbf{Mathematics Subject Classification:} 94A60,  14G50, 11T71

\section{Introduction}
Let $\fqm$  be the finite field with $q^m$ elements, where $q$ is a power of an odd prime $p$ and $m$ is a positive integer. An $[n,k,d]$ linear code $\C$ over $\fq$ is a $k$-dimensional subspace of $\fq^n$ with minimum Hamming distance $d$. An $[n,k,d]$ linear code $\C$ is called optimal (resp. almost optimal) if its parameters $n$, $k$ and $d$ (resp. $d+1$) meet any bound on linear codes \cite{HP}. The weight enumerator of a code $\C$ with length $n$ is the polynomial $1+A_1z+A_2z^2+\cdots+A_nz^n$, where $A_i$ denotes the number of codewords with Hamming weight $i$ in $\C$. The weight distribution $(1,A_1,\ldots,A_n)$ is a significant research topic in coding theory, which contains important information for estimating the error correcting capability and the probability of error detection and correction with respect to some algorithms \cite{KT}. A code $\C$ is said to be a $t$-weight code if the number of nonzero $A_i$ in the sequence $(1,A_1,\dots,A_n)$ is equal to $t$. Linear codes with few weights can be used in secret sharing schemes \cite{ADHK,CDY,YD}, association schemes \cite{CG}, strongly regular graphs \cite{CK} and authentication codes \cite{DW}.

In 2007, Ding and Niederreiter \cite{DN} introduced a nice and generic way to construct linear codes via trace functions. Let $D=\{d_1,d_2,\ldots,d_n\}\subset\fqm$ and define
\begin{align}\label{000}
\C_D=\{(\Tr(bd_1),\Tr(bd_2),\ldots,\Tr(bd_n)):b\in \fqm\},
\end{align}
where $\Tr(\cdot)$ is the trace function from $\fqm$ to $\fq$. Then $\C_D$ is a linear code of length $n$ over $\fq$. The set $D$ is called the defining set of $\C_D$. How to select $D$ such that $\C_D$ has
good parameters is an interesting research problem and many attempts have been made
in this direction to obtain good linear codes, see, for example, \cite{akl,d,hlz,hwlz,kk,ykt,zlfh,DD,WLDL,SQCY,SQRT,MOS}. For more related results, the reader is referred to the survey papers \cite{LS} and \cite{S0}. 
Notably, Heng et al. \cite{HCXL} recently introduced a new defining set to construct linear codes by using the intersection and difference of two sets $\{x\in \mathcal{S}:\Tr(x^2)=t_1\}$ and $\{x\in \mathcal{S}:\Tr(x)=t_2\}$, where
$t_1,t_2\in \fq$ and $\mathcal{S}$ is defined by $\fqm^*=\fq^*\mathcal{S}=\{yz:y\in\fq^*, z\in \mathcal{S}\}$ with $z_1/z_2\notin\fq^*$ for distinct $z_1,z_2\in\mathcal{S}$.
In their work, several new classes of projective two-weight or three-weight linear codes were presented and their parameters and weight distributions were determined completely.

In this paper, inspired by the previous works, we construct $q$-ary linear codes of the form \eqref{000} from the intersection, difference and union of sets $E_1=\{x\in \fqm^*:\Tr(x^{q^e+1})=u\}$ and $E_2=\{x\in \fqm^*:\Tr(x)=v\}$, where $u,v\in \fq$ and $e$ is a positive integer. Combining the sets $E_1$ and $E_2$, we define the following defining sets:
\begin{align}\nonumber\label{001}
{D}_{1}&=E_1 \cap E_2=\{x\in \mathbb{F}^*_{q^m}:{\Tr}(x^{q^e+1})=u,{\Tr}(x)=v\};\\
D_{2}&=E_1 \setminus E_2=\{x\in \mathbb{F}^*_{q^m}:{\Tr}(x^{q^e+1})=u,{\Tr}(x)\neq v\};\\\nonumber
D_{3}&=E_1 \cup E_2=\{x\in \mathbb{F}^*_{q^m}:{\Tr}(x^{q^e+1})=u\,\,{\rm or}\,\,{\Tr}(x)=v\}.
\end{align}
We mainly investigate the codes $\C_{D_1}$, $\C_{D_2}$ and $\C_{D_3}$ defined by \eqref{000} and \eqref{001} in this paper.
Note that if $m/\gcd(m,e)$ is odd, we have $\gcd(q^e+1,q^m-1)=2$ and it leads to $\{x^{q^e+1}:x\in \fqm^*\}=\{x^2:x\in \fqm^*\}$. Thus, for the case that $m/\gcd(m,e)$ is odd, the code $\C_{D_1}$ is the same as the codes in the nonprojective case constructed in \cite{HCXL}, and $\C_{D_2}$ and $\C_{D_3}$ can be studied similarly to the results in \cite{HCXL}. Therefore, we focus on the case that $m/\gcd(m,e)$ is even for our codes $\C_{D_1}$, $\C_{D_2}$ and $\C_{D_3}$.
Through finer calculations on certain exponential sums, we obtain several classes of $t$-weight linear codes over $\fq$ from our constructions, where $t=3,4,5,6$. The parameters and weight distributions of these codes are completely determined by using Weil sums and Gauss sums. Notably, we obtain three classes of optimal codes meeting the Griesmer bound. Moreover, computer experiments using MAGMA programs show that many (almost) optimal linear codes can be derived from our constructions, as shown in Table \ref{tabb}.

The remainder of this paper is organized as follows. Section \ref{s2} introduces some basic concepts and auxiliary results on Gauss sums and Weil sums. Section \ref{s3} presents several classes of linear codes over $\fq$, and the parameters and weight distributions of these codes are determined completely. Section \ref{s4} provides the proofs of our main results. Section \ref{s5} concludes this paper.

\section{Preliminaries}\label{s2}
In this section, we present some preliminaries which will be used to prove our main results. From now on we fix the following notation:\begin{itemize}
  \item  Let $q$ be a power of an odd prime $p$.
  \item  Let $m=2\ell>0$ be an even integer.
  \item  Let $\fqm$ be the finite field with $q^m$ elements and $\fqm^*=\fqm\setminus\{0\}$.
  \item  Let $m_p \in \fp$ such that $m_p  = m \mod p$.
  \item  Let $e$ be a positive integer and $\alpha=\gcd(m,e)$.
  \item  $T=\{b\in \fqm:X^{q^{2e}}+X=-b^{q^e}$ is solvable in $\fqm$\}.
  \item  Let $\gamma\in\fqm$ denote a solution of the equation $X^{q^{2e}}+X=-b^{q^e}$ if $b\in T$.
  \item  Let $\Tr(\cdot)$ (resp. $\Tr^q_{p}(\cdot)$) be the trace function from $\fqm$ to $\fq$ (resp. $\mathbb{F}_{q}$ to $\mathbb{F}_{p}$).
  \item  Let $\chi$ and $\mathbb{\chi}_{1}$ be the canonical additive characters of $\fqm$ and $\fq$,  respectively.
  \item  Let $\eta$ and $\mathbb{\eta}_{1}$ be the quadratic multiplicative characters of $\fqm$ and $\fq$, respectively. We extend these quadratic characters by letting $\eta(0)=0$ and $\mathbb{\eta}_{1}(0)=0$.
\end{itemize}

It's known that for any $x\in\fq^*$ we have$$\eta(x):=\left\{\begin{array}{ll}
1, & \mbox{ if $m$ is even};\\
\eta_1(x), & \mbox{ if $m$ is odd}.
\end{array}\right.$$
The quadratic Gauss sum $G(\eta,\chi)$ over $\fqm$ is defined by $$G(\eta)=G(\eta,\chi)=\sum_{x\in\mathbb{F}_{q^m}^{*}}\mathbb\chi(x)\eta(x).$$
The explicit value of $G(\eta,\chi)$ is given as follows.
\begin{lem}{\rm\cite[Theorem 5.15] {LN}}\label{m1}
Let $q^m=p^s$, where $p$ is an odd prime and $s$ is a positive integer. Then we have $$G(\eta,\chi)=(-1)^{s-1}(-1)^{\frac{s(p-1)^2}{8}}p^{\frac{s}{2}}.$$
\end{lem}
Weil sums are defined by $\sum_{x\in\mathbb{F}_{q^m}}\chi(f(x))$ where $f(x)\in\fqm[x]$. If $f(x)$ is a quadratic polynomial and $q$ is odd, Weil sums have an interesting relationship with quadratic Gauss sums.
\begin{lem}{\rm\cite[Theorem 5.33] {LN}}\label{m2}
Let $f(x)={a}_{2}x^2+{a}_{1}x+{a}_{0}\in\fqm[x]$ with ${a}_{2}\neq0$. Then
$$\sum_{x\in\fqm}\chi(f(x))=\chi({a}_{0}-{a}_{1}^2(4{a}_{2})^{-1})\eta({a}_{2})G(\eta,\chi).$$
\end{lem}

For $a\in\mathbb{F}_{q^m}^{*}, b\in\fqm$, define the Weil sum $S(a,b)$ by $$S(a,b)=\sum_{x\in\fqm}\chi(ax^{q^e+1}+bx).$$
The following lemmas will be used to prove our main results.
\begin{lem}{\rm\cite[Theorem 2] {RSC}}\label{m4}
If ${\frac{m}{\alpha}}$ be even, where $\alpha=\gcd(m,e)$, then $$S(a,0)=\left\{\begin{array}{ll}
q^{\ell}, & {\rm if} \,\, \frac{\ell}{\alpha} \,\,{\rm is\,\, even}\,\, {\rm and}\,\, a^{\frac{q^m-1}{q^\alpha+1}}\neq(-1)^\frac{\ell}{\alpha} ;\\
-q^{\ell+\alpha}, & {\rm if} \,\, \frac{\ell}{\alpha} \,\, {\rm is \,\, even} \,\,  {\rm and}\,\, a^{\frac{q^m-1}{q^\alpha+1}}=(-1)^\frac{\ell}{\alpha} ;\\
-q^{\ell}, & {\rm if} \,\, \frac{\ell}{\alpha} \,\,{\rm is \,\, odd} \,\,{\rm and}\,\, a^{\frac{q^m-1}{q^\alpha+1}}\neq(-1)^\frac{\ell}{\alpha};\\
q^{\ell+\alpha}, & {\rm if} \,\, \frac{\ell}{\alpha} \,\,{\rm is \,\, odd }\,\, {\rm and}\,\, a^{\frac{q^m-1}{q^\alpha+1}}=(-1)^\frac{\ell}{\alpha} .\\
\end{array}\right.$$
\end{lem}
\begin{lem}{\rm\cite[Theorem 4.1] {RSC}}\label{m5}
The equation $$a^{q^e}X^{q^{2e}}+aX=0$$ is solvable over $\mathbb{F}_{q^m}^{*}$ if and only if $\frac{m}{\alpha}$ is even and $a^{\frac{q^m-1}{q^\alpha+1}}=(-1)^\frac{\ell}{\alpha}$, where $\alpha=\gcd(m,e)$. There are $q^{2\alpha}-1$ non-zero solutions.
\end{lem}

From \cite[Theorem 1] {RSCF}, we have the following result on $S(a,b)$ for $\frac{m}{\alpha}$ even.
\begin{lem}{\rm\cite[Theorem 1] {RSCF}}\label{m6}
Suppose $f(X)=a^{q^e}X^{q^{2e}}+aX$ is a permutation polynomial over $\fqm$. Let ${x}_{0}$ be the unique solution of the equation $f(X)=-b^{q^e}$ and $\alpha=\gcd(m,e)$.
If $\frac{m}{\alpha}$ is even, then $a^{\frac{q^m-1}{q^\alpha+1}}\neq(-1)^\frac{\ell}{\alpha}$ and $$S(a,b)=(-1)^{\frac{\ell}{\alpha}}q^{\ell}\chi(-a{x}_{0}^{q^e+1}).$$
\end{lem}
\begin{lem}{\rm\cite[Theorem 2] {RSCF}}\label{m7}
Suppose $f(X)=a^{q^e}X^{q^{2e}}+aX$ is not a permutation polynomial over $\fqm$, then $S(a,b)=0$ unless the equation $f(X)=-b^{q^e}$ is solvable. If the equation $f(X)=-b^{q^e}$ is solvable, with some solution ${x}_{0}$, then $$S(a,b)=-(-1)^{\frac{\ell}{\alpha}}q^{\ell+\alpha}\chi(-a{x}_{0}^{q^e+1}).$$
\end{lem}
\begin{lem}{\rm\cite[Lemma 13] {JLF}}\label{mm1}
If $\frac{\ell}{\alpha}$ is even, where $\alpha=\gcd(m,e)$, then $$|\{b\in\fqm:X^{q^{2e}}+X=b^{q^e}\,\,{\rm is\,\, solvable \,\,in}\,\, \fqm\}|=q^{m-2\alpha}.$$
\end{lem}
\begin{lem}{\rm\cite[Lemma 10] {JLF}}\label{n1}
If $\frac{m}{\alpha}$ is even, where $\alpha=\gcd(m,e)$, then for $x\in\mathbb{F}_{q}^{*}$, $$x^{\frac{q^m-1}{q^\alpha+1}}=1.$$
\end{lem}

\section{Main results} \label{s3}
In this section, we present several classes of linear codes over $\fq$ and determine their parameters and weight distributions. The proofs of these results will be given in Section \ref{s4}.

We first recall the defining sets in \eqref{001} given by the following:
\begin{align}\nonumber
{D}_{1}&=\{x\in \mathbb{F}^*_{q^m}:{\Tr}(x^{q^e+1})=u,{\Tr}(x)=v\};\\  \nonumber
D_{2}&=\{x\in \mathbb{F}^*_{q^m}:{\Tr}(x^{q^e+1})=u,{\Tr}(x)\neq v\};\\ \nonumber
D_{3}&=\{x\in \mathbb{F}^*_{q^m}:{\Tr}(x^{q^e+1})=u\,\,{\rm or}\,\,{\Tr}(x)=v\},
\end{align}
where $u,v\in \fq$. We will focus on the case that $m/\gcd(m,e)$ is even in this paper.

Let $m/\gcd(m,e)$ be even in the sequel. Recalling the $m=2\ell$, define the notation
\begin{equation}\label{e1}
\epsilon:=\left\{\begin{array}{ll}
0, & \mbox{ if $\ell/\gcd(m,e)$ is odd};\\
\gcd(m,e), & \mbox{ if $\ell/\gcd(m,e)$ is even}.
\end{array}\right.
\end{equation}

\subsection{Linear codes from the defining set $D_1$} \label{sub3.1}
In this subsection, we will investigate the linear code $\C_{D_1}$ defined as in \eqref{000} and \eqref{001} by considering two cases: 1) $v=0$; 2) $v\in\fq^*$.
\begin{theorem}\label{t1}
Let $\C_{D_1}$ be defined by \eqref{000} and \eqref{001}, $m/\gcd(m,e)$ be even, $v=0$ and $\epsilon$ be given as in \eqref{e1}. Assume that $m>3$, and $m>2\epsilon+2$ if $m_p\neq0$. \\
1) If $u=0$ and $m_p\neq0$, then ${\C}_{{D}_{1}}$ is a $3$-weight $[{q^{m-2}-1},m-1,(q-1)(q^{m-3}-q^{\ell+\epsilon-2})]$ linear code over $\fq$ with the weight distribution
\begin{center}\footnotesize
\begin{tabular}{cc}
\hline
\rm Weight  & \rm Frequency                                     \\ \hline
$0$              & $1$               \\
$(q-1)q^{m-3}$              &$q^{m-1}-(q-1)q^{m-2\epsilon-2}-1$             \\
$(q-1)(q^{m-3}+q^{\ell+\epsilon-2})$                      &$\frac{1}{2}(q-1)(q^{m-2\epsilon-2}-q^{\ell-\epsilon-1})$      \\
$(q-1)(q^{m-3}-q^{\ell+\epsilon-2})$              &$\frac{1}{2}(q-1)(q^{m-2\epsilon-2}+q^{\ell-\epsilon-1})$               \\ \hline
\end{tabular}
\end{center}
2) If $u=0$ and $m_p=0$, then ${\C}_{{D}_{1}}$ is a $3$-weight $[q^{m-2}-(q-1)q^{\ell+\epsilon-1}-1,m-1,(q-1)(q^{m-3}-q^{\ell+\epsilon-1})]$ linear code over $\fq$ with the weight distribution
\begin{center}\footnotesize
\begin{tabular}{cc}\hline
 \rm Weight  &  \rm Frequency                                     \\ \hline
$0$              & $1$               \\
$(q-1)q^{m-3}$              &$q^{m-2\epsilon-3}-(q-1)q^{\ell-\epsilon-2}-1$             \\
$(q-1)(q^{m-3}-q^{\ell+\epsilon-1})$                      &$(q-1)(q^{m-2\epsilon-3}+q^{\ell-\epsilon-2})$      \\
$(q-1)(q^{m-3}+q^{\ell+\epsilon-2}-q^{\ell+\epsilon-1})$              &$q^{m-1}-q^{m-2\epsilon-2}$               \\ \hline
\end{tabular}
\end{center}
3) If $u\neq0$ and $m_p\neq0$, then ${\C}_{{D}_{1}}$ is a $3$-weight $[q^{m-2}-q^{\ell+\epsilon-1}{\eta}_{1}(-um_p),m-1]$ linear code over $\fq$ with the weight distribution
      \begin{center}\footnotesize
\begin{tabular}{cc}
\hline
 \rm Weight  &  \rm Frequency                                     \\ \hline
$0$              & $1$               \\
$(q-1)q^{m-3}$              &$q^{m-2\epsilon-2}-1$             \\
\multirow{2}{*}{$(q-1)q^{m-3}-q^{\ell+\epsilon-1}{\eta}_{1}(-um_p)+ q^{\ell+\epsilon-2}$}                      &$\frac{1}{2}(q^{m-1}-q^{m-2\epsilon-2}-q^{m-2\epsilon-1}{\eta}_{1}(-um_p)$\\&$+q^{m-1}{\eta}_{1}(-um_p)+q^{\ell-\epsilon}-q^{\ell-\epsilon-1})$  \\
\multirow{2}{*}{$(q-1)q^{m-3}-q^{\ell+\epsilon-1}{\eta}_{1}(-um_p)-q^{\ell+\epsilon-2}$}              &$\frac{1}{2}(q^{m-1}-q^{m-2\epsilon-2}+q^{m-2\epsilon-1}{\eta}_{1}(-um_p)$\\&$-q^{m-1}{\eta}_{1}(-um_p)-q^{\ell-\epsilon}+q^{\ell-\epsilon-1})$  \\ \hline
\end{tabular}
\end{center}
4) If $u\neq0$ and $m_p=0$, then ${\C}_{{D}_{1}}$ is a $3$-weight $[q^{m-2}+q^{\ell+\epsilon-1},m-1,(q-1)q^{m-3}]$ linear code over $\fq$ with the weight distribution
      \begin{center}\footnotesize
\begin{tabular}{cc}
\hline
 \rm Weight  &  \rm Frequency                                     \\ \hline
$0$              & $1$               \\
$(q-1)q^{m-3}$              &$q^{m-2\epsilon-2}-\frac{1}{2}(q-1)(q^{m-2\epsilon-3}+q^{\ell-\epsilon-2})-1$             \\
$(q-1)q^{m-3}+2q^{\ell+\epsilon-1}$                      &$\frac{1}{2}(q-1)(q^{m-2\epsilon-3}+q^{\ell-\epsilon-2})$      \\
$(q-1)(q^{m-3}+q^{\ell+\epsilon-2})$              &$q^{m-1}-q^{m-2\epsilon-2}$              \\ \hline
\end{tabular}
\end{center}
\end{theorem}

\begin{example}
 Let $(q,m,e)=(9,4,2)$ and $(u,v)=(0,0)$. Magma experiments show that ${\C}_{{D}_{1}}$ is a $[80,3,64]$ linear code over $\mathbb{F}_{9}$ with the weight enumerator $1+360z^{64}+80z^{72}+288z^{80}$, which is consistent with our result in Theorem \ref{t1}.
\end{example}
\begin{example}
 Let $(q,m,e)=(3,6,1)$ and $(u,v)=(0,0)$. Magma experiments show that ${\C}_{{D}_{1}}$ is a $[62,5,36]$ linear code over $\mathbb{F}_{3}$ with the weight enumerator $1+60z^{36}+162z^{42}+20z^{54}$, which is consistent with our result in Theorem \ref{t1}.
\end{example}
\begin{example}
Let $(q,m,e)=(9,4,2)$ and $(u,v)=(1,0)$. Magma experiments show that ${\C}_{{D}_{1}}$ is a $[72,3,62]$ linear code over $\mathbb{F}_{9}$ with the weight enumerator $1+288z^{62}+360z^{64}+80z^{72}$, which is consistent with our result in Theorem \ref{t1}. This code is almost optimal due to \cite{CT}.
\end{example}
\begin{example}
Let $(q,m,e)=(3,12,1)$ and $(u,v)=(1,0)$. Magma experiments show that ${\C}_{{D}_{1}}$ is a $[59778,11,39366]$ linear code over $\mathbb{F}_{3}$ with the weight enumerator $1+4346z^{39366}+170586z^{39852}+2214z^{40824}$, which is consistent with our result in Theorem \ref{t1}.
\end{example}

Notice that in Theorem \ref{t1}, when $m/\gcd(m,e)$ is even, $m_p\neq0$ and $m>3$, it is clear that $m=2\epsilon+2$ if and only if $m=4$ and $\epsilon=1$ (i.e., $e=1,3$). For this case, we present two classes of optimal linear codes over $\fq$ in the following. We omit the proof since it can be proved similarly to Theorem \ref{t1}.

\begin{cor}\label{b1}
Let $\C_{D_1}$ be defined by \eqref{000} and \eqref{001} and $\epsilon$ be given as in \eqref{e1}. Assume that $m=4$ and $\epsilon=1$ (i.e., $e=1,3$). \\
1) If $u=0$, then $\C_{D_1}$ is a $1$-weight $[q^2-1,2,q^2-q]$ optimal linear code meeting the Griesmer bound with the weight enumerator $1+(q^2-1)z^{(q-1)q}$.\\
2) If $u\neq0$ and $\eta_1(-um_p)=-1$, then $\C_{D_1}$ is a $2$-weight $[2q^2,3,2q^2-2q]$ optimal linear code meeting the Griesmer bound with the weight enumerator $1+(q^3-q)z^{2q(q-1)}+(q-1)z^{2q^2}$.
\end{cor}

\begin{example}
Let $(q,m,e)=(5,4,1)$ and $(u,v)=(0,0)$. Magma experiments show that ${\C}_{{D}_{1}}$ is a $[24,2,20]$ linear code over $\mathbb{F}_{5}$ with the weight enumerator $1+24z^{20}$, which is consistent with our result in Corollary \ref{b1}. This code is optimal due to \cite{CT}.
\end{example}
\begin{example}
Let $(q,m,e)=(3,4,1)$ and $(u,v)=(1,0)$. Magma experiments show that ${\C}_{{D}_{1}}$ is a $[18,3,12]$ linear code over $\mathbb{F}_{3}$ with the weight enumerator $1+24z^{12}+2z^{18}$, which is consistent with our result in Corollary \ref{b1}. This code is optimal due to \cite{CT}.
\end{example}

\begin{theorem}\label{t7}
Let $\C_{D_1}$ be defined by \eqref{000} and \eqref{001}, $m/\gcd(m,e)$ be even, $v\in\fq^*$ and $\epsilon$ be given as in \eqref{e1}. Assume that $m>3$, and $m>2\epsilon+2$ if $m_p\neq0$. \\
1) If $v^2-um_p\neq0$ and $m_p\neq0$, then ${\C}_{{D}_{1}}$ is a $6$-weight $[q^{m-2}-q^{\ell+\epsilon-1}{\eta}_{1}(v^2-um_p),m]$ linear code over $\fq$ with the weight distribution
\begin{center} \footnotesize
\begin{tabular}{cc}
\hline
 \rm Weight  &  \rm Frequency                                     \\ \hline
$0$              & $1$               \\
$(q-1)q^{m-3}$              &$q^{m-2\epsilon-2}-1$             \\
$(q-1)q^{m-3}-q^{\ell+\epsilon-2}{\eta}_{1}(v^2-um_p)$                      &$(q-1)(q^{m-2\epsilon-2}-q^{\ell-\epsilon-1}{\eta}_{1}(v^2-um_p)$      \\
$(q-1)q^{m-3}-q^{\ell+\epsilon-1}{\eta}_{1}(v^2-um_p)$              &$(q-1)(q^{m-2\epsilon-2}-1)$               \\
$q^{m-2}-q^{\ell+\epsilon-1}{\eta}_{1}(v^2-um_p)$              &$q-1$               \\
\multirow{2}{*}{$(q-1)q^{m-3}-q^{\ell+\epsilon-1}{\eta}_{1}(v^2-um_p)+q^{\ell+\epsilon-2}$} &$\frac{q-1}{2}(({\eta}_{1}(v^2-um_p)-1)(q^{\ell-\epsilon-1}+q^{m-2\epsilon-2})+q^{\ell-\epsilon})+\frac{1}{2}$\\&$(q^m({\eta}_{1}(v^2-um_p)+1)-q^{m-2\epsilon-1}-q^{m-2\epsilon}\eta_1(v^2-um_p))$    \\
\multirow{2}{*}{$(q-1)q^{m-3}-q^{\ell+\epsilon-1}{\eta}_{1}(v^2-um_p)-q^{\ell+\epsilon-2}$}   &$\frac{q-1}{2}(({\eta}_{1}(v^2-um_p)+1)(q^{\ell-\epsilon-1}-q^{m-2\epsilon-2})-q^{\ell-\epsilon})+\frac{1}{2}$\\&$(q^m(1-\eta_1(v^2-um_p))+q^{m-2\epsilon}\eta_1(v^2-um_p)-q^{m-2\epsilon-1})$      \\ \hline
\end{tabular}
\end{center}
2) If $m_p=0$, then ${\C}_{{D}_{1}}$ is a $4$-weight $[q^{m-2},m,(q-1)q^{m-3}-q^{\ell+\epsilon-2}]$ linear code over $\fq$ with the weight distribution
      \begin{center}\footnotesize
\begin{tabular}{cc}
\hline
\rm Weight  & \rm Frequency                                     \\ \hline
$0$              & $1$               \\
$(q-1)q^{m-3}$              &$q^m-q^{m-2\epsilon}+q^{m-2\epsilon-1}-q$             \\
$q^{m-2}$              &$q-1$               \\
$(q-1)q^{m-3}-q^{\ell+\epsilon-2}$              &$(q-1)(q^{m-2\epsilon-1}-q^{m-2\epsilon-2})$               \\
$(q-1)(q^{m-3}+q^{\ell+\epsilon-2})$              &$(q-1)q^{m-2\epsilon-2}$               \\ \hline
\end{tabular}
\end{center}
3) If $v^2-um_p=0$, then ${\C}_{{D}_{1}}$ is a $6$-weight $[q^{m-2},m,(q-1)(q^{m-3}-q^{\ell+\epsilon-2})]$ linear code over $\fq$ with the weight distribution
\begin{center}\footnotesize
\begin{tabular}{cc}
\hline
\rm Weight  & \rm Frequency                                     \\ \hline
$0$              & $1$               \\
$(q-1)q^{m-3}$              &$q^m-q^{m-2\epsilon}+q^{m-2\epsilon-1}-q$             \\
$(q-1)q^{m-3}-q^{\ell+\epsilon-2}$                      &$\frac{(q-1)^2}{2}(q^{m-2\epsilon-2}-q^{\ell-\epsilon-1})$      \\
$(q-1)q^{m-3}+q^{\ell+\epsilon-2}$              &$\frac{(q-1)^2}{2}(q^{m-2\epsilon-2}+q^{\ell-\epsilon-1})$               \\
$q^{m-2}$              &$q-1$               \\
$(q-1)(q^{m-3}+q^{\ell+\epsilon-2})$              &$\frac{q-1}{2}(q^{m-2\epsilon-2}-q^{\ell-\epsilon-1})$               \\
$(q-1)(q^{m-3}-q^{\ell+\epsilon-2})$              &$\frac{q-1}{2}(q^{m-2\epsilon-2}+q^{\ell-\epsilon-1})$               \\ \hline
\end{tabular}
\end{center}
\end{theorem}

\begin{example}
Let $(q,m,e)=(9,4,2)$ and $(u,v)=(0,2)$. Magma experiments show that ${\C}_{{D}_{1}}$ is a $[72,4,62]$ linear code over $\mathbb{F}_{9}$ with the weight enumerator $1+2016z^{62}+640z^{63}+3240z^{64}+576z^{71}+88z^{72}$, which is consistent with our result in Theorem \ref{t7}. This code is optimal due to \cite{CT}.
\end{example}
\begin{example}
Let $(q,m,e)=(3,6,1)$ and $(u,v)=(0,2)$. Magma experiments show that ${\C}_{{D}_{1}}$ is a $[81,6,51]$ linear code over $\mathbb{F}_{3}$ with the weight enumerator $1+324z^{51}+240z^{54}+162z^{60}+2z^{81}$, which is consistent with our result in Theorem \ref{t7}. This code is optimal due to \cite{CT}.
\end{example}
\begin{example}
Let $(q,m,e)=(3,4,2)$ and $(u,v)=(1,1)$. Magma experiments show that ${\C}_{{D}_{1}}$ is a $[9,4,4]$ linear code over $\mathbb{F}_{3}$ with the weight enumerator $1+12z^{4}+12z^{5}+24z^{6}+24z^{7}+6z^{8}+2z^{9}$, which is consistent with our result in Theorem \ref{t7}. This code is almost optimal due to \cite{CT}.
\end{example}


Similar to Corollary \ref{b1}, we have the following result.
\begin{cor}\label{b3}
Let $\C_{D_1}$ be defined by \eqref{000} and \eqref{001} and $\epsilon$ be given as in \eqref{e1}. Assume that $m=4$ and $\epsilon=1$ (i.e., $e=1,3$). \\
1) If $v^2-um_p\neq0$ and $\eta_1(-um_p)=-1$, then $\C_{D_1}$ is a $3$-weight $[2q^2,4,q^2]$ linear code with the weight enumerator $1+2(q-1)z^{q^2}+(q^4-q^2)z^{2q(q-1)}+(q-1)^2z^{2q^2}$.\\
2) If $v^2-u m_p=0$, then  $\C_{D_1}$ is a $2$-weight $[q^2,3,q^2-q]$ optimal linear code meeting the Griesmer bound with the weight enumerator $1+(q^3-q)z^{(q-1)q}+(q-1)z^{q^2}$.
\end{cor}

\begin{example}
Let $(q,m,e)=(5,4,1)$ and $(u,v)=(1,1)$. Magma experiments show that ${\C}_{{D}_{1}}$ is a $[50,4,25]$ linear code over $\mathbb{F}_{5}$ with the weight enumerator $1+8z^{25}+600z^{40}+16z^{50}$, which is consistent with our result in Corollary \ref{b3}.
According to \cite{CT}, the minimum distance of the best known linear codes over ${\mathbb F}_{5}$ with length $50$ and dimension $4$ is $38$.
\end{example}
\begin{example}
Let $(q,m,e)=(5,4,1)$ and $(u,v)=(1,2)$. Magma experiments show that ${\C}_{{D}_{1}}$ is a $[25,3,20]$ linear code over $\mathbb{F}_{5}$ with the weight enumerator $1+120z^{20}+4z^{25}$, which is consistent with our result in Corollary \ref{b3}. This code is almost optimal due to \cite{CT}.
\end{example}
\subsection{Linear codes from the defining set $D_2$}  \label{sub3.2}
In this subsection, we will investigate the linear code $\C_{D_2}$ defined as in \eqref{000} and \eqref{001} by considering two cases: 1) $v=0$; 2) $v\in\fq^*$.

\begin{theorem}\label{t4}
Let $\C_{D_2}$ be defined by \eqref{000} and \eqref{001}, $m/\gcd(m,e)$ be even, $u=0$ , $v=0$ and $\epsilon$ be given as in \eqref{e1}. Assume that $m>3$, and $m>2\epsilon+2$ if $m_p\neq0$. \\
1) If $m_p\neq0$, then ${\C}_{{D}_{2}}$ is a $6$-weight $[(q-1)(q^{m-2}-q^{\ell+\epsilon-1}),m,(q-1)(q^{m-2}-q^{m-3}-q^{\ell+\epsilon-1}-q^{\ell+\epsilon-2})]$ linear code over $\fq$ with the weight distribution
      \begin{center}\footnotesize
\begin{tabular}{cc}
\hline
 \rm Weight  &  \rm Frequency                                     \\ \hline
$0$              & $1$               \\
$(q-1)(q^{m-2}-q^{\ell+\epsilon-1})$              &$q-1$             \\
$(q-1)(q^{m-2}-q^{m-3})$                      &$q^{m-2\epsilon-2}-1$      \\
$(q-1)(q^{m-2}-q^{m-3}-q^{\ell+\epsilon-2})$              &$(q-1)(q^{m-2\epsilon-2}-q^{\ell-\epsilon-1})$        \\
$(q-1)(q^{m-2}-q^{m-3}-q^{\ell+\epsilon-1})$                      &$(q-1)(q^{m-2\epsilon-2}-1)$      \\
$(q-1)(q^{m-2}-q^{m-3}-q^{\ell+\epsilon-1}-q^{\ell+\epsilon-2})$                      &$\frac{q-1}{2}(q^{m-2\epsilon-1}-q^{\ell-\epsilon}+2q^{\ell-\epsilon-1}-2q^{m-2\epsilon-2})$      \\
$(q-1)^2(q^{m-3}-q^{\ell+\epsilon-2})$                      &$q^m+\frac{1}{2}(q^{\ell-\epsilon+1}-q^{m-2\epsilon}-q^{m-2\epsilon-1}-q^{\ell-\epsilon})$      \\ \hline
\end{tabular}
\end{center}
2) If $m_p=0$, then ${\C}_{{D}_{2}}$ is a $4$-weight $[(q-1)q^{m-2},m,(q-1)(q^{m-2}-q^{m-3}-q^{\ell+\epsilon-2})]$ linear code over $\fq$ with the weight distribution
      \begin{center}\footnotesize
\begin{tabular}{cc}
\hline
 \rm Weight  &  \rm Frequency                                     \\ \hline
$0$              & $1$               \\
$(q-1)q^{m-2}$              &$q-1$             \\
$(q-1)^2(q^{m-3}+q^{\ell+\epsilon-2})$                      &$(q-1)q^{m-2\epsilon-2}$      \\
$(q-1)(q^{m-2}-q^{m-3}-q^{\ell+\epsilon-2})$              &$(q-1)^2q^{m-2\epsilon-2}$               \\
$(q-1)(q^{m-2}-q^{m-3})$                      &$q^m-q^{m-2\epsilon}+q^{m-2\epsilon-1}-q$      \\ \hline
\end{tabular}
\end{center}
\end{theorem}

\begin{example}
Let $(q,m,e)=(3,8,1)$ and $(u,v)=(0,0)$. Magma experiments show that ${\C}_{{D}_{2}}$ is a $[1296,8,756]$ linear code over $\mathbb{F}_{3}$ with the weight enumerator $1+72z^{756}+160z^{810}+6102z^{864}+144z^{918}+80z^{972}+2z^{1296}$, which is consistent with our result in Theorem \ref{t4}.
\end{example}
\begin{example}
Let $(q,m,e)=(3,6,1)$ and $(u,v)=(0,0)$. Magma experiments show that ${\C}_{{D}_{2}}$ is a $[162,6,102]$ linear code over $\mathbb{F}_{3}$ with the weight enumerator $1+324z^{102}+240z^{108}+162z^{120}+2z^{162}$, which is consistent with our result in Theorem \ref{t4}.
\end{example}

\begin{theorem}\label{t10}
Let $\C_{D_2}$ be defined by \eqref{000} and \eqref{001}, $m/\gcd(m,e)$ be even, $u=0$, $v\in\fq^*$ and $\epsilon$ be given as in \eqref{e1}. Assume that $m>3$, and $m>2\epsilon+2$ if $m_p\neq0$. \\
1) If $m_p\neq0$, then ${\C}_{{D}_{2}}$ is a $6$-weight $[(q-1)(q^{m-2}-q^{\ell+\epsilon-1})+q^{\ell+\epsilon-1}-1,m,(q-1)(q^{m-2}-q^{\ell+\epsilon-1})-q^{m-2}+q^{\ell+\epsilon-1}]$ linear code over $\fq$ with the weight distribution
      \begin{center}\footnotesize
\begin{tabular}{cc}
\hline
 \rm Weight  &  \rm Frequency                                     \\ \hline
$0$              & $1$               \\
$(q-1)(q^{m-2}-q^{\ell+\epsilon-1})-q^{m-2}+q^{\ell+\epsilon-1}$              &$q-1$             \\
$(q-1)^2(q^{m-3}-q^{\ell+\epsilon-2})$                      &$\frac{q-1}{2}(q^{\ell-\epsilon}-q^{m-2\epsilon-1})-q^{m-2\epsilon-1}+q^{m-2\epsilon}$      \\
$(q-1)(q^{m-2}-q^{m-3}-q^{\ell+\epsilon-1})+q^{\ell+\epsilon-1}$              &$(q-1)(q^{m-2\epsilon-2}-1)$        \\
$(q-1)(q^{m-2}-q^{m-3}-q^{\ell+\epsilon-1})+q^{\ell+\epsilon-1}+q^{\ell+\epsilon-2}$                      &$\frac{q-1}{2}(q^{m-2\epsilon-1}-2q^{m-2\epsilon-2}-q^{\ell-\epsilon}+2q^{\ell-\epsilon-1})+q^m-q^{m-2\epsilon}$      \\
$(q-1)^2q^{m-3}$                      &$q^{m-2\epsilon-2}-1$      \\
$(q-1)^2q^{m-3}+q^{\ell+\epsilon-2}$                      &$(q-1)(q^{m-2\epsilon-2}-q^{\ell-\epsilon-1})$      \\ \hline
\end{tabular}
\end{center}
2) If $m_p=0$, then ${\C}_{{D}_{2}}$ is a $6$-weight $[(q-1)(q^{m-2}-q^{\ell+\epsilon-1})-1,m,(q-1)q^{m-2}-q^{m-2}]$ linear code over $\fq$ with the weight distribution
\begin{center}\footnotesize
\begin{tabular}{cc}
\hline
 \rm Weight  &  \rm Frequency                                     \\ \hline
$0$              & $1$               \\
$(q-1)^2(q^{m-3}-q^{\ell+\epsilon-2})$              &$q^m-q^{m-2\epsilon}$             \\
$(q-1)q^{m-2}-q^{m-2}$                      &$q-1$      \\
$(q-1)(q^{m-2}-q^{m-3}-q^{\ell+\epsilon-2})$              &$(q-1)q^{m-2\epsilon-2}$               \\
$(q-1)(q^{m-2}-q^{m-3})$                      &$q^{m-2\epsilon-2}-(q-1)q^{\ell-\epsilon-1}-q$      \\
$(q-1)(q^{m-2}-q^{m-3}-q^{\ell+\epsilon-1})$                      &$(q-1)(q^{m-2\epsilon-2}+q^{\ell-\epsilon-1})$      \\
$(q-1)(q^{m-2}-q^{m-3}-q^{\ell+\epsilon-1})+q^{\ell+\epsilon-2}$                      &$q^{m-2\epsilon}-2q^{m-2\epsilon-1}+q^{m-2\epsilon-2}$      \\ \hline
\end{tabular}
\end{center}
\end{theorem}

\begin{example}
Let $(q,m,e)=(3,8,1)$ and $(u,v)=(0,1)$. Magma experiments show that ${\C}_{{D}_{2}}$ is a $[1376,8,648]$ linear code over $\mathbb{F}_{3}$ with the weight enumerator $1+2z^{648}+270z^{864}+160z^{891}+5904z^{918}+80z^{972}+144z^{999}$, which is consistent with our result in Theorem \ref{t10}.
\end{example}
\begin{example}
Let $(q,m,e)=(3,6,1)$ and $(u,v)=(0,1)$. Magma experiments show that ${\C}_{{D}_{2}}$ is a $[143,6,81]$ linear code over $\mathbb{F}_{3}$ with the weight enumerator $1+2z^{81}+180z^{90}+324z^{93}+162z^{102}+60z^{108}$, which is consistent with our result in Theorem \ref{t10}.
\end{example}

\begin{remark}
When $m_p=0$ and $\ell/\gcd(m,e)$ is odd, the code $\C_{D_2}$ in 2) of Theorem \ref{t10} is reduced to a $5$-weight linear code.
\end{remark}

\begin{remark} \label{remark-1}
Note that we focus solely on the case $u=0$ for the code $\C_{D_2}$ in Theorems \ref{t4} and \ref{t10}, due to the computational difficulties associated with the weight distribution of the code $\C_{D_2}$ when $u\ne 0$. It remains a problem to study the code $\C_{D_2}$ for the case $u\ne 0$ when $m/\gcd(m,e)$ is even, which is $t$-weight with $6\leq t\leq8$ for $q=5$ and $5\leq m<10$ by Magma experiments.
\end{remark}


\subsection{Linear codes from the defining set $D_3$}  \label{sub3.3}
In this subsection, we will investigate the linear code $\C_{D_3}$ defined as in \eqref{000} and \eqref{001} by considering two cases: 1) $v=0$; 2) $v\in\fq^*$.

\begin{theorem}\label{t5}
Let $\C_{D_3}$ be defined by \eqref{000} and \eqref{001}, $m/\gcd(m,e)$ be even, $u=0$, $v=0$ and $\epsilon$ be given as in \eqref{e1}. Assume that $m>3$, and $m>2\epsilon+2$ if $m_p\neq0$.\\
1) If $m_p\neq0$, then ${\C}_{{D}_{3}}$ is a $6$-weight $[(q-1)(q^{m-2}-q^{\ell+\epsilon-1})+q^{m-1}-1,m,(q-1)(q^{m-2}-q^{\ell+\epsilon-1})]$ linear code over $\fq$  with the weight distribution
      \begin{center} \footnotesize
\begin{tabular}{cc}
\hline
\rm Weight  & \rm Frequency                                     \\ \hline
$0$              & $1$               \\
$(q-1)(q^{m-2}-q^{\ell+\epsilon-1})$              &$q-1$             \\
$(q-1)(2q^{m-2}-q^{m-3})$                      &$q^{m-2\epsilon-2}-1$      \\
$(q-1)(2q^{m-2}-q^{m-3}-q^{\ell+\epsilon-2})$              &$(q-1)(q^{m-2\epsilon-2}-q^{\ell-\epsilon-1})$               \\
$(q-1)(2q^{m-2}-q^{m-3}-q^{\ell+\epsilon-1})$                      &$(q-1)(q^{m-2\epsilon-2}-1)$      \\
$(q-1)(2q^{m-2}-q^{m-3}-q^{\ell+\epsilon-1}-q^{\ell+\epsilon-2})$                      &$\frac{q-1}{2}(q^{m-2\epsilon-1}-q^{\ell-\epsilon}+2q^{\ell-\epsilon-1}-2q^{m-2\epsilon-2})$      \\
$(q-1)(2q^{m-2}-q^{m-3}-q^{\ell+\epsilon-1}+q^{\ell+\epsilon-2})$                      &$q^m+\frac{1}{2}(q^{\ell-\epsilon+1}-q^{m-2\epsilon}-q^{m-2\epsilon-1}-q^{\ell-\epsilon})$      \\ \hline
\end{tabular}
\end{center}
2) If $m_p=0$, then ${\C}_{{D}_{3}}$ is a $4$-weight $[(q-1)q^{m-2}+q^{m-1}-1,m,(q-1)q^{m-2}]$ linear code over $\fq$  with the weight distribution
      \begin{center}\footnotesize
\begin{tabular}{cc}
\hline
\rm Weight  & \rm Frequency                                     \\ \hline
$0$              & $1$               \\
$(q-1)q^{m-2}$              &$q-1$             \\
$(q-1)(2q^{m-2}-q^{m-3}+q^{\ell+\epsilon-1}-q^{\ell+\epsilon-2})$                      &$(q-1)q^{m-2\epsilon-2}$      \\
$(q-1)(2q^{m-2}-q^{m-3}-q^{\ell+\epsilon-2})$              &$(q-1)^2q^{m-2\epsilon-2}$               \\
$(q-1)(2q^{m-2}-q^{m-3})$                      &$q^m-q^{m-2\epsilon}+q^{m-2\epsilon-1}-q$      \\  \hline
\end{tabular}
\end{center}
\end{theorem}
\begin{example}
Let $(q,m,e)=(3,4,2)$ and $(u,v)=(0,0)$. Magma experiments show that ${\C}_{{D}_{3}}$ is a $[38,4,12]$ linear code over $\mathbb{F}_{3}$ with the weight enumerator $1+2z^{12}+6z^{22}+16z^{24}+36z^{26}+12z^{28}+8z^{30}$, which is consistent with our result in Theorem \ref{t5}.
\end{example}
\begin{example}
Let $(q,m,e)=(3,12,1)$ and $(u,v)=(0,0)$. Magma experiments show that ${\C}_{{D}_{3}}$ is a $[295244,12,118098]$ linear code over $\mathbb{F}_{3}$ with the weight enumerator $1+2z^{118098}+26244z^{196344}+492072z^{196830}+13122z^{197802}$, which is consistent with our result in Theorem \ref{t5}.
\end{example}

\begin{theorem}\label{t11}
Let $\C_{D_3}$ be defined by \eqref{000} and \eqref{001}, $m/\gcd(m,e)$ be even, $u=0$, $v\in\fq^*$ and $\epsilon$ be given as in \eqref{e1}. Assume that $m>3$, and $m>2\epsilon+2$ if $m_p\neq0$.\\
1) If $m_p\neq0$, then ${\C}_{{D}_{3}}$ is a $6$-weight $[(q-1)(q^{m-2}-q^{\ell+\epsilon-1})+q^{m-1}+q^{\ell+\epsilon-1}-1,m,(q-1)^2(q^{m-3}-q^{\ell+\epsilon-2})+(q-1)q^{m-2}]$ linear code over $\fq$ with the weight distribution
 \begin{center}\footnotesize
\begin{tabular}{cc}
\hline
 \rm Weight  &  \rm Frequency                                     \\ \hline
$0$              & $1$               \\
$(q-1)(2q^{m-2}-q^{\ell+\epsilon-1})+q^{\ell+\epsilon-1}$              &$q-1$             \\
$(q-1)^2(q^{m-3}-q^{\ell+\epsilon-2})+(q-1)q^{m-2}$                      &$\frac{q-1}{2}(q^{\ell-\epsilon}-q^{m-2\epsilon-1})-q^{m-2\epsilon-1}+q^{m-2\epsilon}$      \\
$(q-1)(2q^{m-2}-q^{m-3}-q^{\ell+\epsilon-1})+q^{\ell+\epsilon-1}$     &$(q-1)(q^{m-2\epsilon-2}-1)$        \\
$(q-1)(2q^{m-2}-\,q^{m-3}\,-q^{\ell+\epsilon-1})+(q+1)q^{\ell+\epsilon-2}$                      &$\frac{q-1}{2}(q^{m-2\epsilon-1}-2q^{m-2\epsilon-2}-q^{\ell-\epsilon}+2q^{\ell-\epsilon-1})+q^m-q^{m-2\epsilon}$     \\
$(q-1)^2q^{m-3}+(q-1)q^{m-2}$                      &$q^{m-2\epsilon-2}-1$      \\
$(q-1)(2q^{m-2}-q^{m-3})+q^{\ell+\epsilon-2}$       &$(q-1)(q^{m-2\epsilon-2}-q^{\ell-\epsilon-1})$      \\  \hline
\end{tabular}
\end{center}
2) If $m_p=0$, then ${\C}_{{D}_{3}}$ is a $6$-weight $[(q-1)(q^{m-2}-q^{\ell+\epsilon-1})+q^{m-1}-1,m,(q-1)(2q^{m-2}-q^{m-3}-q^{\ell+\epsilon-1}]$ linear code over $\fq$ with the weight distribution
\begin{center}\footnotesize
\begin{tabular}{cc}
\hline
 \rm Weight  &  \rm Frequency                                     \\ \hline
$0$              & $1$               \\
$(q-1)^2(q^{m-3}-q^{\ell+\epsilon-2})+(q-1)q^{m-2}$              &$q^m-q^{m-2\epsilon}$             \\
$2(q-1)q^{m-2}$                      &$q-1$      \\
$(q-1)(2q^{m-2}-q^{m-3}-q^{\ell+\epsilon-2})$              &$(q-1)q^{m-2\epsilon-2}$               \\
$(q-1)(2q^{m-2}-q^{m-3})$                      &$q^{m-2\epsilon-2}-(q-1)q^{\ell-\epsilon-1}-q$      \\
$(q-1)(2q^{m-2}-q^{m-3}-q^{\ell+\epsilon-1})$                      &$(q-1)(q^{m-2\epsilon-2}+q^{\ell-\epsilon-1})$      \\
$(q-1)(2q^{m-2}-q^{m-3}-q^{\ell+\epsilon-1})+q^{\ell+\epsilon-2}$                      &$q^{m-2\epsilon}-2q^{m-2\epsilon-1}+q^{m-2\epsilon-2}$      \\ \hline
\end{tabular}
\end{center}
\end{theorem}

\begin{example}
Let $(q,m,e)=(3,8,1)$ and $(u,v)=(0,1)$. Magma experiments show that ${\C}_{{D}_{3}}$ is a $[3563,8,2322]$ linear code over $\mathbb{F}_{3}$ with the weight enumerator $1+270z^{2322}+160z^{2349}+5904z^{2376}+80z^{2430}+144z^{2457}+2z^{2835}$, which is consistent with our result in Theorem \ref{t11}.
\end{example}

\begin{example}
Let $(q,m,e)=(3,6,1)$ and $(u,v)=(0,1)$. Magma experiments show that ${\C}_{{D}_{3}}$ is a $[386,6,252]$ linear code over $\mathbb{F}_{3}$ with the weight enumerator $1+180z^{252}+324z^{255}+162z^{264}+60z^{270}+2z^{324}$, which is consistent with our result in Theorem \ref{t11}.
\end{example}

\begin{remark}
When $m_p=0$ and $\ell/\gcd(m,e)$ is odd, the code $\C_{D_3}$ in 2) of Theorem \ref{t11} is reduced to a $5$-weight linear code.
\end{remark}


\begin{remark}
Note that we also consider only the case $u=0$ for the code $\C_{D_3}$ in Theorems \ref{t5} and \ref{t11}, due to the same reasons as for the code $\C_{D_2}$ (see remark \ref{remark-1}). It also remains a problem to study the code $\C_{D_3}$ for the case $u\ne 0$ when $m/\gcd(m,e)$ is even, which is $t$-weight with $6\leq t\leq8$ for $q=5$ and $5\leq m<10$ by Magma experiments.
\end{remark}



In the following remark, we provide a comparison of our codes to the previous works.

\begin{remark}
In general, it is difficult to discuss the equivalence of codes. It is well-known
that equivalent codes have the same parameters and weight distribution, but the converse
is not necessarily true. Some interesting linear codes with few weights were presented in \cite{HCXL,DD,CTF,WLDL,zlfh,ms,ykt,TLQ,hs}. By comparing
our codes with those in the above references, we find that when $\ell/\gcd(m,e)$ is odd, the parameters of our codes in 1) of Theorem \ref{t1} are the same as those in \cite[Theorem 1] {DD}; the parameters of our codes in 1) and 2) of Theorem \ref{t1} are the same as those in \cite[Theorem 1] {ms}.
Notably, the other codes in this paper are different from the known ones in \cite{HCXL,DD,CTF,WLDL,zlfh,ms,ykt,TLQ,hs}.
\end{remark}

\section{Proofs of main results} \label{s4}
In this section, we give the proofs of our main results. To this end, we first provide some auxiliary results.

\subsection{Some auxiliary lemmas}
In this subsection, we give some lemmas to compute the parameters and weight distributions of the linear codes defined by \eqref{000} and \eqref{001}.

The following lemma will be used to compute the length of our codes.
\begin{lem}\label{lem1}
Let $u,v\in \fq$, $m/\gcd(m,e)$ be even and $\epsilon$ be given as in \eqref{e1}. Define
\begin{align}\label{n2}
N:=|\{x\in\fqm:{\Tr}(x^{q^e+1})=u,{\Tr}(x)=v\}|.
\end{align}
1) If $u=v=0$, then
  \begin{align}
  N&=\left\{\begin{array}{ll}
q^{m-2}, & {\rm if} \,\, m_p\neq0;\\
q^{m-2}-(q-1)q^{\ell+\epsilon-1}, & {\rm if} \,\,  m_p=0.
\end{array}\right.\nonumber
  \end{align}
2) If $u\neq0,v=0$, then
  \begin{align}
  N&=\left\{\begin{array}{ll}
q^{m-2}-q^{\ell+\epsilon-1}{\eta}_{1}(-um_p), & {\rm if}  \,\, m_p\neq0;\\
q^{m-2}+q^{\ell+\epsilon-1}, & {\rm if} \,\, m_p=0.
\end{array}\right.\nonumber
  \end{align}
3) If $v\neq0$, then
  \begin{align}
  N&=\left\{\begin{array}{ll}
q^{m-2}-q^{\ell+\epsilon-1}{\eta}_{1}(v^2-um_p), & {\rm if} \,\,  m_p\neq0,v^2-um_p\neq0;\\
q^{m-2}, & {\rm if} \,\,m_p=0\,\,{\rm or}\,\,v^2-um_p=0.
\end{array}\right.\nonumber
  \end{align}
\end{lem}

\begin{proof}
By the orthogonal relation of additive character, we have
\begin{align}\label{002}
N=&\frac{1}{q^2}\sum_{x\in\mathbb{F}_{q^m} }\sum_{y\in\mathbb{F}_{q}}\chi_1(y(\Tr(x^{q^{e}+1})-u))\sum_{z\in\mathbb{F}_{q}}
\chi_1(z(\Tr(x)-v))
\nonumber\\
=&\frac{1}{q^2}\sum_{x\in\mathbb{F}_{q^m} }\sum_{y\in\mathbb{F}_{q}}\chi(yx^{q^{e}+1})\sum_{z\in\mathbb{F}_{q}}\chi(zx)\mathbb{\chi}_{1}
(-uy-vz)\nonumber\\
=&\frac{1}{q^2}\sum_{x\in\mathbb{F}_{q^m} }\sum_{z\in\mathbb{F}_{q}}\chi(zx)\mathbb{\chi}_{1}(-vz)+\frac{1}{q^2}\sum_{x\in\mathbb{F}_{q^m} }\sum_{y\in\mathbb{F}_{q}^{*}}\sum_{z\in\mathbb{F}_{q}}\chi(yx^{q^{e}+1}+zx)\mathbb{\chi}_{1}(-uy-vz)\nonumber\\
=&q^{m-2}+\frac{1}{q^2}\Theta,
\end{align}
where
$$\Theta:=\sum_{y\in\mathbb{F}_{q}^{*}}\sum_{z\in\mathbb{F}_{q}}\mathbb{\chi}_{1}(-uy-vz)\sum_{x\in\fqm }\chi(yx^{q^{e}+1}+zx).$$
Observe that $S(y,z)=\sum_{x\in\fqm }\chi(yx^{q^{e}+1}+zx)$.
Note that for $y\in\fq^*$ and $z\in\fq$, $\frac{-z}{2y}$ is the solution of $y^{q^e}X^{q^{2e}}+yX=-z^{q^e}$. By Lemmas \ref{m6}, \ref{m7} and \ref{n1}, it leads to $$S(y,z)=-q^{\ell+\epsilon}\chi(-y(\frac{-z}{2y})^{q^e+1})=-q^{\ell+\epsilon}\chi(-\frac{z^2}{4y})
=-q^{\ell+\epsilon}\chi_{1}(-\frac{mz^2}{4y}).$$
This together with Lemma \ref{m2} gives
\begin{align}
\Theta&=\sum_{y\in\mathbb{F}_{q}^{*}}\sum_{z\in\mathbb{F}_{q}}-q^{\ell+\epsilon}\mathbb{\chi}_{1}(\frac{-m_pz^2}{4y}-uy-vz)\nonumber\\
&=\left\{\begin{array}{ll}
-q^{\ell+\epsilon}\sum\limits_{y\in\mathbb{F}_{q}^{*}}{G}_{1}({\eta}_{1})
\mathbb{\chi}_{1}(-uy+\frac{v^2y}{m_p}){\eta}_{1}(\frac{-m_p}{4y}), & {\rm if}  \,\, m_p\neq0;\\
-q^{\ell+\epsilon}\sum\limits_{y\in\mathbb{F}_{q}^{*}}\sum\limits_{z\in\mathbb{F}_{q}}\mathbb{\chi}_{1}(-uy-vz), & {\rm if}\,\, m_p=0.
\end{array}\right.\nonumber
\end{align}

Next we evaluate $\Theta$ by consider the following three cases.\\
1) If $u=v=0$, then
\begin{align*}
  \Theta&=\left\{\begin{array}{ll}
-q^{\ell+\epsilon}\sum\limits_{y\in\mathbb{F}_{q}^{*}}{G}_{1}({\eta}_{1}){\eta}_{1}(\frac{-m_p}{4y}), & {\rm if} \,\, m_p\neq0;\\
\sum\limits_{y\in\mathbb{F}_{q}^{*}}\sum\limits_{z\in\mathbb{F}_{q}}-q^{\ell+\epsilon}, & {\rm if}\,\, m_p=0;
\end{array}\right.\\
&=\left\{\begin{array}{ll}
0, & {\rm if} \,\, m_p\neq0;\\
-(q-1)q^{\ell+\epsilon+1}, & {\rm if}\,\, m_p=0.
\end{array}\right.\nonumber
  \end{align*}
2) If $u\neq0,v=0$, due to the facts that ${G}^2_{1}({\eta}_{1})={\eta}_{1}(-1)q$ and ${\eta}_{1}(\frac{-m_p}{4y})=\eta_1((-u+\frac{v^2}{m_p})y)\eta_1(-u+\frac{v^2}{m_p})\eta_1(-\frac{m_p}{4})$, we have
\begin{align}
  \Theta&=\left\{\begin{array}{ll}
-q^{\ell+\epsilon}{G}^2_{1}({\eta}_{1}){\eta}_{1}(um_p), & {\rm if} \,\, m_p\neq0;\\
-q^{\ell+\epsilon}\sum\limits_{y\in\mathbb{F}_{q}^{*}}\sum\limits_{z\in\mathbb{F}_{q}}\mathbb{\chi}_{1}(-uy), & {\rm if}\,\, m_p=0;
\end{array}\right.\nonumber\\
&=\left\{\begin{array}{ll}
-q^{\ell+\epsilon+1}{\eta}_{1}(-um_p), & {\rm if} \,\, m_p\neq0;\\
q^{\ell+\epsilon+1}, & {\rm if}\,\, m_p=0.
\end{array}\right.\nonumber
\end{align}
3) If $v\neq0$, similar to the proof of 2), it gives
\begin{align}
\Theta&=\left\{\begin{array}{ll}
-q^{\ell+\epsilon}{G}^2_{1}({\eta}_{1}){\eta}_{1}(um_p-v^2), & {\rm if}  \,\, m_p\neq0,v^2-um_p\neq0;\\
-q^{\ell+\epsilon}\sum\limits_{y\in\mathbb{F}_{q}^{*}}{G}_{1}({\eta}_{1}){\eta}_{1}(\frac{-m_p}{4y}),
& {\rm if}  \,\, m_p\neq0,v^2-um_p=0;\\
0, & {\rm if}\,\, m_p=0;
\end{array}\right.\nonumber\\
&=\left\{\begin{array}{ll}
-q^{\ell+\epsilon+1}{\eta}_{1}(v^2-um_p), & {\rm if} \,\, m_p\neq0,v^2-um_p\neq0;\\
0, & {\rm if} \,\,  m_p=0\,\, {\rm or}\,\,v^2-um_p=0.
\end{array}\right.\nonumber
  \end{align}

By (\ref{002}), this completes the proof.
\end{proof}

To compute the weight distributions of our codes, we need to compute the following exponential sum  \begin{align}\label{nn}
\Theta_b(u,v):=\sum_{x\in\fqm}\sum_{{y}_{1}\in\mathbb{F}_{q}^{*}}\sum_{{y}_{2},{y}_{3}\in
\mathbb{F}_{q}}\chi({y}_{1}x^{q^e+1}+({y}_{2}+{y}_{3}b)x)\mathbb{\chi}_{1}(-u{y}_{1}-v{y}_{2}),
\end{align}
where $b\in \fqm^*$.

Next we compute the values of exponential sums $\Theta_b(u,v)$.
\begin{lem}\label{n3}
Let $b\in\mathbb{F}_{q^m}^*$, $m/\gcd(m,e)$ be even and $\epsilon$ be given as in \eqref{e1}. Recall the notation $T$ and $\gamma$ as in the beginning of Section \ref{s2}. Then the value of ${\Theta}_{b}(0,0)$ in \eqref{nn} is given in the following\\
1) If $m_p\neq0$, then $${\Theta}_{b}(0,0)=\left\{\begin{array}{ll}
0, & {\rm if} \,\, b\not\in T \,\,{\rm or}\,\,b\in T,A=0;\\
-(q-1)q^{\ell+\epsilon+1}, & {\rm if}\,\,b\in T,{\Tr}(\gamma^{q^u+1})=0,{\Tr}({\gamma})\neq0;\\
-(q-1)q^{\ell+\epsilon+1}{\eta}_{1}(-m_p{\Tr}(\gamma^{q^e+1})), & {\rm if}\,\,b\in T,{\Tr}(\gamma^{q^e+1})\neq0,{\Tr}({\gamma})=0;\\
-(q-1)q^{\ell+\epsilon+1}{\eta}_{1}({\Tr}(\gamma)^2-m_p{\Tr}(\gamma^{q^e+1})),&{\rm if}\,\,b\in T,{\Tr}(\gamma^{q^e+1})\neq0,{\Tr}({\gamma})\neq0, \\&\,\,\,\, A\neq0.
\end{array}\right.$$
2)  If $m_p=0$, then $${\Theta}_{b}(0,0)=\left\{\begin{array}{ll}
-(q-1)q^{\ell+\epsilon+2}, & {\rm if} \,\, b\in T,{\Tr}(\gamma^{q^e+1})=0,{\Tr}(\gamma)=0;\\
-(q-1)q^{\ell+\epsilon+1}, & {\rm if} \,\,b\in T,{\Tr}({\gamma})\neq0\,\,{\rm or}\,\,b\not\in T;\\
0, & {\rm if} \,\, b\in T,{\Tr}(\gamma^{q^e+1})\neq0,{\Tr}({\gamma})=0;
\end{array}\right.$$
where $A={\Tr}(\gamma)^2-m_p{\Tr}(\gamma^{q^e+1})$.
\end{lem}
\begin{proof}
At first, we prove that the equation $y^{q^e}_1X^{q^{2e}}+y_1X=-(y_2+y_3b)^{q^e}$ is insolvable if $b\notin T$, where $y_2\in \fq$ and $y_1,y_3\in \fq^*$. Suppose that  $y^{q^e}_1X^{q^{2e}}+y_1X=-(y_2+y_3b)^{q^e}$ has a solution $\varepsilon$ when $b\notin T$. Then $\varepsilon+\frac{y_2}{2y_1}$ is the solution of $y^{q^e}_1X^{q^{2e}}+y_1X=-(y_3b)^{q^e}$ since $\frac{y_2}{2y_1}$ is a solution of $y^{q^e}_1X^{q^{2e}}+y_1X=(y_2)^{q^e}$. Accordingly, it gives that $\varepsilon y_1y_3^{-1}+\frac{y_2}{2y_3}$ is the solution of $X^{q^{2e}}+X=-(b)^{q^e}$, which implies $b\in T$. This is a contradiction.
Therefore, the equation $y^{q^e}_1X^{q^{2e}}+y_1X=-(y_2+y_3b)^{q^e}$ is insolvable if $b\notin T$.
It then follows from Lemma \ref{m7} that for $b\notin T$, we have
$$\begin{aligned}
\Theta_b(0,0)&=\sum_{x\in\fqm}\sum_{{y}_{1},{y}_{3}\in\mathbb{F}_{q}^{*}}\sum_{{y}_{2}\in\mathbb{F}_{q}}\chi({y}_{1}x^{q^e+1}+({y}_{2}+{y}_{3}b)x)+\sum_{x\in\mathbb{F}_{q^m}}\sum_{{y}_{1}\in\mathbb{F}_{q}^{*}}\sum_{{y}_{2}\in\mathbb{F}_{q}}\chi({y}_{1}x^{q^e+1}+{y}_{2}x)\\
&=\sum_{{y}_{1}\in\mathbb{F}_{q}^{*}}\sum_{{y}_{2}\in\mathbb{F}_{q}}\sum_{x\in\mathbb{F}_{q^m}}
\chi({y}_{1}x^{q^e+1}+{y}_{2}x).\nonumber
\end{aligned}$$

When $b\in T$, for $y_1\in\fq^*$ and $y_2,y_3\in\fq$, $y^{-1}_1y_3\gamma$ is the solution of $y^{q^e}_1X^{q^{2e}}+y_1X=-(y_3b)^{q^e}$ and $-\frac{1}{2}y^{-1}_1y_2$ is the solution of $y^{q^e}_1X^{q^{2e}}+y_1X=-(y_2)^{q^e}$. Then $y^{-1}_1(y_3\gamma-\frac{1}{2}y_2)$ is the solution of $y^{q^e}_1X^{q^{2e}}+y_1X=-(y_2+y_3b)^{q^e}$.
By Lemmas \ref{m6}, \ref{m7} and \ref{n1}, for $b\in T$, we have $${\Theta}_{b}(0,0)=
\sum_{{y}_{1}\in\mathbb{F}_{q}^{*}}\sum_{{y}_{2},{y}_{3}\in\mathbb{F}_{q}}
-q^{\ell+\epsilon}\chi(-{y}_{1}({y}_{1}^{-1}({y}_{3}\gamma-\frac{1}{2}{y}_{2}))^{q^e+1}).$$

Next we further compute ${\Theta}_{b}(0,0)$ by considering the following two cases.

1) $m_p\neq0$. For $b\not\in T$, it follows from Lemmas \ref{m6}, \ref{m7}, \ref{n1} and \ref{m2} that
$$\begin{aligned}\label{0c1}
{\Theta}_{b}(0,0)=&\sum_{x\in\mathbb{F}_{q^m}}\sum_{{y}_{1}\in\mathbb{F}_{q}^{*}}\sum_{{y}_{2}\in\mathbb{F}_{q}}\chi({y}_{1}x^{q^e+1}+{y}_{2}x)=-q^{\ell+\epsilon}\sum_{{y}_{1}\in\mathbb{F}_{q}^{*}}\sum_{{y}_{2}\in\mathbb{F}_{q}}\mathbb{\chi}_{1}(\frac{-m_p{y}^2_{2}}{4{y}_{1}})\nonumber\\
=&-q^{\ell+\epsilon}\sum_{{y}_{1}\in\mathbb{F}_{q}^{*}}{G}_{1}({\eta}_{1}){\eta}_{1}(\frac{-m_p}{4{y}_{1}})=0.\nonumber\\
\end{aligned}$$
 For $b\in T$, it  gives
\begin{align*}
{\Theta}_{b}(0,0)=&\sum_{{y}_{1}\in\mathbb{F}_{q}^{*}}\sum_{{y}_{2},{y}_{3}\in\mathbb{F}_{q}}-q^{\ell+\epsilon}\chi(-{y}_{1}({y}_{1}^{-1}({y}_{3}\gamma-\frac{1}{2}{y}_{2}))^{q^e+1})\\
=&\sum_{{y}_{1}\in\mathbb{F}_{q}^{*}}\sum_{{y}_{2},{y}_{3}\in\mathbb{F}_{q}}-q^{\ell+\epsilon}\mathbb{\chi}_{1}(\frac{-1}{{y}_{1}}\Tr(({y}_{3}\gamma-\frac{1}{2}{y}_{2})^{q^{e}+1})).
\end{align*}
It's known that $\Tr(({y}_{3}\gamma-\frac{1}{2}{y}_{2})^{q^{e}+1})=\Tr(y^2_3\gamma^{q^{e}+1}+\frac{y^2_2}{4}-\frac{1}{2}y_2y_3
\gamma^{q^e}-\frac{1}{2}y_2y_3\gamma)={y}^2_{3}{\Tr}(\gamma^{q^e+1})-{y}_{2}{y}_{3}{\Tr}(\gamma)+
\frac{m_p{y}^2_{2}}{4}$. Then the value of ${\Theta}_{b}(0,0)$ is equal to
 \begin{eqnarray}
&&\sum\limits_{{y}_{1}\in\mathbb{F}_{q}^{*}}\sum\limits_{{y}_{2},{y}_{3}\in\mathbb{F}_{q}}-q^{\ell+\epsilon}\mathbb{\chi}_{1}(\frac{-1}{{y}_{1}}({y}^2_{3}{\Tr}(\gamma^{q^e+1})-{y}_{2}{y}_{3}{\Tr}(\gamma)+\frac{m_p{y}^2_{2}}{4}))\nonumber\\\nonumber
&=&\left\{\begin{array}{ll}
-q^{\ell+\epsilon}\sum\limits_{{y}_{1}\in\mathbb{F}_{q}^{*}}\sum\limits_{{y}_{2},{y}_{3}\in\mathbb{F}_{q}}\mathbb{\chi}_{1}(\frac{-m_p{y}^2_{2}}{4{y}_{1}}), & {\rm if} \,\, {\Tr}(\gamma^{q^e+1})=0,{\Tr}(\gamma)=0;\\
-q^{\ell+\epsilon}\sum\limits_{{y}_{1}\in\mathbb{F}_{q}^{*}}\sum\limits_{{y}_{2},{y}_{3}\in\mathbb{F}_{q}}\mathbb{\chi}_{1}(\frac{{y}_{2}{y}_{3}{\Tr}(\gamma)}{{y}_{1}}-\frac{m_p{y}^2_{2}}{4{y}_{1}}), & {\rm if} \,\,{\Tr}(\gamma^{q^e+1})=0,{\Tr}({\gamma})\neq0;\\
-q^{\ell+\epsilon}\sum\limits_{{y}_{1}\in\mathbb{F}_{q}^{*}}\sum\limits_{{y}_{2},{y}_{3}\in\mathbb{F}_{q}}\mathbb{\chi}_{1}(\frac{-{y}^2_{3}{\Tr}(\gamma^{q^e+1})}{{y}_{1}}-\frac{m_p{y}^2_{2}}{4{y}_{1}}), & {\rm if} \,\,{\Tr}(\gamma^{q^e+1})\neq0,{\Tr}({\gamma})=0;\\
-q^{\ell+\epsilon}\sum\limits_{{y}_{1}\in\mathbb{F}_{q}^{*}}\sum\limits_{{y}_{2},{y}_{3}\in\mathbb{F}_{q}}\mathbb{\chi}_{1}(\frac{-{y}^2_{3}{\Tr}(\gamma^{q^e+1})}{{y}_{1}}+\frac{{y}_{2}{y}_{3}{\Tr}(\gamma)}{{y}_{1}}-\frac{m_p{y}^2_{2}}{4{y}_{1}}), &{\rm if}\,\,{\Tr}(\gamma^{q^e+1})\neq0,{\Tr}({\gamma})\neq0.
\end{array}\right.\nonumber
\end{eqnarray}
Let $A={\Tr}(\gamma)^2-m_p{\Tr}(\gamma^{q^e+1})$. Then $A=0$ implies ${\Tr}(\gamma^{q^e+1})=0,{\Tr}({\gamma})=0$ or ${\Tr}(\gamma^{q^e+1})\neq0,{\Tr}({\gamma})\neq0$. By Lemma \ref{m2}, the value of ${\Theta}_{b}(0,0)$ is equal to
\begin{eqnarray}\nonumber
& &\left\{\begin{array}{ll}
-q^{\ell+\epsilon}\sum\limits_{{y}_{1}\in\mathbb{F}_{q}^{*}}\sum\limits_{{y}_{3}\in\mathbb{F}_{q}}{G}_{1}({\eta}_{1}){\eta}_{1}(\frac{-m_p}{4{y}_{1}}), & {\rm if} \,\, {\Tr}(\gamma^{q^e+1})=0,{\Tr}(\gamma)=0; \\
-q^{\ell+\epsilon}\sum\limits_{{y}_{1}\in\mathbb{F}_{q}^{*}}\sum\limits_{{y}_{3}\in\mathbb{F}_{q}}{G}_{1}({\eta}_{1}){\eta}_{1}(\frac{-m_p}{4{y}_{1}})\mathbb{\chi}_{1}(\frac{{y}^2_{3}{{\Tr}(\gamma)}^2}{m_p{y}_{1}}), & {\rm if} \,\,{\Tr}(\gamma^{q^e+1})=0,{\Tr}({\gamma})\neq0; \\
-q^{\ell+\epsilon}\sum\limits_{{y}_{1}\in\mathbb{F}_{q}^{*}}{G}^2_{1}({\eta}_{1}){\eta}_{1}(\frac{-{\Tr}(\gamma^{q^e+1})}{{y}_{1}}){\eta}_{1}(\frac{-m_p}{4{y}_{1}}), & {\rm if} \,\,{\Tr}(\gamma^{q^e+1})\neq0,{\Tr}({\gamma})=0; \\
-q^{\ell+\epsilon+1}\sum\limits_{{y}_{1}\in\mathbb{F}_{q}^{*}}{G}_{1}({\eta}_{1}){\eta}_{1}(\frac{-m_p}{4{y}_{1}}), & {\rm if}\,\,{\Tr}(\gamma^{q^e+1})\neq0,{\Tr}({\gamma})\neq0,A=0; \\
-q^{\ell+\epsilon}\sum\limits_{{y}_{1}\in\mathbb{F}_{q}^{*}}\sum\limits_{{y}_{3}\in\mathbb{F}_{q}}{G}_{1}({\eta}_{1}){\eta}_{1}(\frac{-m_p}{4{y}_{1}})\mathbb{\chi}_{1}(\frac{{y}^2_{3}A}{m_p{y}_{1}}), &{\rm if} \,\,{\Tr}(\gamma^{q^e+1})\neq0,{\Tr}({\gamma})\neq0 ,A\neq0.
\end{array}\right.\\\nonumber
\end{eqnarray}
This together with the fact ${G}^2_{1}({\eta}_{1})={\eta}_{1}(-1)q$ gives
\begin{eqnarray}\nonumber
{\Theta}_{b}(0,0)&=&\left\{\begin{array}{ll}
0, & {\rm if} \,\, A=0;\\
-q^{\ell+\epsilon}\sum\limits_{{y}_{1}\in\mathbb{F}_{q}^{*}}{G}^2_{1}({\eta}_{1}){\eta}_{1}(m_p{{y}_{1}}){\eta}_{1}(\frac{-m_p}{4{y}_{1}}), & {\rm if} \,\,{\Tr}(\gamma^{q^e+1})=0,{\Tr}({\gamma})\neq0;\\
-q^{\ell+\epsilon+1}\sum\limits_{{y}_{1}\in\mathbb{F}_{q}^{*}}{\eta}_{1}(-m_p{\Tr}(\gamma^{q^e+1})), & {\rm if} \,\,{\Tr}(\gamma^{q^e+1})\neq0,{\Tr}({\gamma})=0;\\
-q^{\ell+\epsilon}\sum\limits_{{y}_{1}\in\mathbb{F}_{q}^{*}}{G}^2_{1}({\eta}_{1}){\eta}_{1}(\frac{-m_p}{4{y}_{1}}){\eta}_{1}(\frac{A}{m_p{y}_{1}}), & {\rm if} \,\,{\Tr}(\gamma^{q^e+1})\neq0,{\Tr}({\gamma})\neq0, A\neq0.
\end{array}\right.\\\nonumber
\end{eqnarray}
A direct computation leads to
\begin{eqnarray}\nonumber
&=&\left\{\begin{array}{ll}
0, & {\rm if} \,\,A=0;\\
-(q-1)q^{\ell+\epsilon+1}, & {\rm if}\,\,{\Tr}(\gamma^{q^u+1})=0,{\Tr}({\gamma})\neq0;\\
-(q-1)q^{\ell+\epsilon+1}{\eta}_{1}(-m_p{\Tr}(\gamma^{q^e+1})), & {\rm if}\, \,{\Tr}(\gamma^{q^e+1})\neq0,{\Tr}({\gamma})=0;\\
-(q-1)q^{\ell+\epsilon+1}{\eta}_{1}({\Tr}(\gamma)^2-m_p{\Tr}(\gamma^{q^e+1})),&{\rm if}\,\, {\Tr}(\gamma^{q^e+1})\neq0,{\Tr}({\gamma})\neq0, A\neq0.
\end{array}\right.\nonumber
\end{eqnarray}

2) $m_p=0$. For $b\not\in T$, by Lemmas \ref{m6}, \ref{m7} and \ref{n1}, we have $$\begin{aligned}\label{0c1}
{\Theta}_{b}(0,0)=&\sum_{x\in\mathbb{F}_{q^m}}\sum_{{y}_{1}\in\mathbb{F}_{q}^{*}}\sum_{{y}_{2}\in\mathbb{F}_{q}}\chi({y}_{1}x^{q^e+1}+{y}_{2}x)\nonumber\\
=&-q^{\ell+\epsilon}\sum_{{y}_{1}\in\mathbb{F}_{q}^{*}}\sum_{{y}_{2}\in\mathbb{F}_{q}}\mathbb{\chi}_{1}
(\frac{-m_p{y}^2_{2}}{4{y}_{1}})=-(q-1)q^{\ell+\epsilon+1}.\nonumber\\
\end{aligned}$$
For $b\in T$, by Lemma \ref{m2} and the fact ${G}^2_{1}({\eta}_{1})={\eta}_{1}(-1)q$, we have
\begin{eqnarray}\nonumber
& &{\Theta}_{b}(0,0)=\sum_{{y}_{1}\in\mathbb{F}_{q}^{*}}\sum_{{y}_{2},{y}_{3}\in\mathbb{F}_{q}}-q^{\ell+\epsilon}\mathbb{\chi}_{1}
(\frac{-1}{{y}_{1}}({y}^2_{3}{\Tr}(\gamma^{q^e+1})-{y}_{2}{y}_{3}{\Tr}(\gamma)))\\
 &=&\left\{\begin{array}{ll}
-(q-1)q^{\ell+\epsilon+2}, & {\rm if}\, {\Tr}(\gamma^{q^e+1})=0,{\Tr}(\gamma)=0;\\
-q^{\ell+\epsilon}\sum\limits_{{y}_{1}\in\mathbb{F}_{q}^{*}}\sum\limits_{{y}_{2},{y}_{3}\in\mathbb{F}_{q}}\mathbb{\chi}_{1}(\frac{{y}_{2}{y}_{3}{\Tr}(\gamma)}{{y}_{1}}), &{\rm if}\,\,{\Tr}(\gamma^{q^e+1})=0,{\Tr}({\gamma})\neq0;\\
-q^{\ell+\epsilon}\sum\limits_{{y}_{1}\in\mathbb{F}_{q}^{*}}\sum\limits_{{y}_{2},{y}_{3}\in\mathbb{F}_{q}}\mathbb{\chi}_{1}(\frac{-{y}^2_{3}{\Tr}(\gamma^{q^e+1})}{{y}_{1}}), & {\rm if}\,\, {\Tr}(\gamma^{q^e+1})\neq0,{\Tr}({\gamma})=0;\\
-q^{\ell+\epsilon}\sum\limits_{{y}_{1}\in\mathbb{F}_{q}^{*}}\sum\limits_{{y}_{2}\in\mathbb{F}_{q}}{G}_{1}({\eta}_{1}){\eta}_{1}(\frac{-{\Tr}(\gamma^{q^e+1})}{{y}_{1}})\mathbb{\chi}_{1}(\frac{{y}^2_{2}{{\Tr}(\gamma)}^2}{4{y}_{1}{\Tr}(\gamma^{q^e+1})}), &{\rm if}\,\,{\Tr}(\gamma^{q^e+1})\neq0,{\Tr}({\gamma})\neq0;
\end{array}\right.\nonumber  \\
 &=& \left\{\begin{array}{ll}
-(q-1)q^{\ell+\epsilon+2}, & {\rm if} \,\, {\Tr}(\gamma^{q^e+1})=0,{\Tr}(\gamma)=0;\\
-(q-1)q^{\ell+\epsilon+1}, & {\rm if} \,\,{\Tr}(\gamma^{q^e+1})=0,{\Tr}({\gamma})\neq0;\\
-q^{\ell+\epsilon+1}\sum_{{y}_{1}\in\mathbb{F}_{q}^{*}}{G}_{1}({\eta}_{1}){\eta}_{1}(\frac{-{\Tr}(\gamma^{q^e+1})}{{y}_{1}}), & {\rm if} \,\, {\Tr}(\gamma^{q^e+1})\neq0,{\Tr}({\gamma})=0;\\
-q^{\ell+\epsilon}\sum_{{y}_{1}\in\mathbb{F}_{q}^{*}}{G}^2_{1}({\eta}_{1}){\eta}_{1}(-1), & {\rm if} \,\, {\Tr}(\gamma^{q^e+1})\neq0,{\Tr}({\gamma})\neq0.
\end{array}\right. \nonumber
\end{eqnarray}
This completes the proof.
\end{proof}

Similar to the computation of ${\Theta}_{b}(0,0)$ in Lemma \ref{n3}, the value of ${\Theta}_{b}(u,0)$, ${\Theta}_{b}(0,v)$ and ${\Theta}_{b}(u,v)$ can be given as in the following lemmas, where $u,v\in \fq^*$.
\begin{lem}\label{n5}
Let $b\in\mathbb{F}_{q^m}^*$, $u\in\fq^*$, $m/\gcd(m,e)$ be even and $\epsilon$ be given as in \eqref{e1}. Recall the notation $T$ and $\gamma$ as in the beginning of Section \ref{s2}. Then the value of ${\Theta}_{b}(u,0)$ in \eqref{nn} is given in the following\\
1)  If $m_p\neq0$, then$${\Theta}_{b}(u,0)=\left\{\begin{array}{ll}
-q^{\ell+\epsilon+1}{\eta}_{1}(-um_p), & {\rm if} \,\, b\not\in T;\\
-q^{\ell+\epsilon+2}{\eta}_{1}(-um_p), & {\rm if} \,\, b\in T, {\Tr}(\gamma)^2-m_p{\Tr}(\gamma^{q^e+1})=0;\\
q^{\ell+\epsilon+1}, & {\rm if} \,\, b\in T,{\Tr}(\gamma^{q^u+1})=0,{\Tr}({\gamma})\neq0;\\
q^{\ell+\epsilon+1}{\eta}_{1}(-m_p{\Tr}(\gamma^{q^e+1})), & {\rm if} \,\,  b\in T,{\Tr}(\gamma^{q^e+1})\neq0,{\Tr}({\gamma})=0;\\
q^{\ell+\epsilon+1}{\eta}_{1}({\Tr}(\gamma)^2-m_p{\Tr}(\gamma^{q^e+1})), & {\rm if} \,\,  b\in T, {\Tr}(\gamma^{q^e+1})\neq0,{\Tr}({\gamma})\neq0, \\& \,\,\,\,{\Tr}(\gamma)^2-m_p{\Tr}(\gamma^{q^e+1})\neq0.
\end{array}\right.$$
2) If $m_p=0$, then $${\Theta}_{b}(u,0)=\left\{\begin{array}{ll}
q^{\ell+\epsilon+2}, & {\rm if} \,\, b\in T,{\Tr}(\gamma)=0,{\Tr}(\gamma^{q^e+1})=0;\\
q^{\ell+\epsilon+1}, & {\rm if} \,\,b\in T,{\Tr}({\gamma})\neq0\,\,{\rm or}\,\,b\not\in T;\\
-q^{\ell+\epsilon+2}{\eta}_{1}(-u{\Tr}(\gamma^{q^e+1})), & {\rm if} \,\, b\in T,{\Tr}(\gamma^{q^e+1})\neq0,{\Tr}({\gamma})=0.
\end{array}\right.$$

\end{lem}
\begin{lem}
Let $b\in\mathbb{F}_{q^m}^*$, $v\in\fq^*$, $m/\gcd(m,e)$ be even and $\epsilon$ be given as in \eqref{e1}. Recall the notation $T$ and $\gamma$ as in the beginning of Section \ref{s2}. Then the value of ${\Theta}_{b}(0,v)$ in \eqref{nn} is given in the following:\\
1) If $m_p\neq0$, then$${\Theta}_{b}(0,v)=\left\{\begin{array}{ll}
-q^{\ell+\epsilon+1}, & {\rm if} \,\,b\not\in T;\\
-q^{\ell+\epsilon+2}, & {\rm if} \,\,b\in T, {\Tr}(\gamma^{q^e+1})=0,{\Tr}(\gamma)=0;\\
-(q-1)q^{\ell+\epsilon+1}, & {\rm if} \,\,b\in T,{\Tr}(\gamma^{q^e+1})=0,{\Tr}({\gamma})\neq0;\\
q^{\ell+\epsilon+1}{\eta}_{1}(-m_p{\Tr}(\gamma^{q^e+1})), & {\rm if} \,\,b\in T, {\Tr}(\gamma^{q^e+1})\neq0,{\Tr}({\gamma})=0;\\
0, & {\rm if} \,\, b\in T, {\Tr}(\gamma^{q^e+1})\neq0,{\Tr}({\gamma})\neq0, \\& \,\,\,\, {\Tr}(\gamma)^2-m_p{\Tr}(\gamma^{q^e+1})=0;\\
q^{\ell+\epsilon+1}{\eta}_{1}({\Tr}(\gamma)^2-m_p{\Tr}(\gamma^{q^e+1})), & {\rm if} \,\,b\in T, {\Tr}(\gamma^{q^e+1})\neq0,{\Tr}({\gamma})\neq0, \\& \,\,\,\,{\Tr}(\gamma)^2-m_p{\Tr}(\gamma^{q^e+1})\neq0.
\end{array}\right.$$
2) If $m_p=0$, then$${\Theta}_{b}(0,v)=\left\{\begin{array}{ll}
0, & {\rm if} \,\,b\in T, {\Tr}(\gamma)=0 \,\,{\rm or}\,\, b\not\in T;\\
-(q-1)q^{\ell+\epsilon+1}, & {\rm if} \,\,b\in T,{\Tr}(\gamma^{q^e+1})=0,{\Tr}({\gamma})\neq0;\\
q^{\ell+\epsilon+1}, & {\rm if} \,\,b\in T,{\Tr}(\gamma^{q^e+1})\neq0,{\Tr}({\gamma})\neq0.
\end{array}\right.$$

\end{lem}
\begin{lem} \label{n10}
Let $b\in\mathbb{F}_{q^m}^*$, $u,v\neq0$, $v^2-um_p\neq0$, $m/\gcd(m,e)$ be even and $\epsilon$ be given as in \eqref{e1}. Recall the notation $T$ and $\gamma$ as in the beginning of Section \ref{s2}. Then the value of ${\Theta}_{b}(u,v)$ in \eqref{nn} is given in the following:\\
1)  If $m_p\neq0$, then$${\Theta}_{b}(u,v)=\left\{\begin{array}{ll}
-q^{\ell+\epsilon+1}{\eta}_{1}(v^2-um_p), & {\rm if} \,\,b\not\in T;\\
-q^{\ell+\epsilon+2}{\eta}_{1}(v^2-um_p), & {\rm if} \,\, b\in T, {\Tr}(\gamma)=0,{\Tr}(\gamma^{q^e+1})=0;\\
q^{\ell+\epsilon+1}, & {\rm if} \,\,b\in T,{\Tr}(\gamma^{q^e+1})=0,{\Tr}({\gamma})\neq0;\\
q^{\ell+\epsilon+1}{\eta}_{1}(-m_p{\Tr}(\gamma^{q^e+1})), & {\rm if} \,\,b\in T,{\Tr}(\gamma^{q^e+1})\neq0,{\Tr}({\gamma})=0;\\
0, & {\rm if} \,\,b\in T,{\Tr}(\gamma^{q^e+1})\neq0,{\Tr}({\gamma})\neq0,A=0;\\
q^{\ell+\epsilon+1}{\eta}_{1}({\Tr}(\gamma)^2-m_p{\Tr}(\gamma^{q^e+1})), & {\rm if} \,\, b\in T,{\Tr}(\gamma^{q^e+1})\neq0,{\Tr}({\gamma})\neq0, \\&\,\,\,\,A\neq0,B\neq0;\\
-(q-1)q^{\ell+\epsilon+1}{\eta}_{1}(v^2-um_p), & {\rm if} \,\, b\in T,{\Tr}(\gamma^{q^e+1})\neq0,{\Tr}({\gamma})\neq0, \\&\,\,\,\,A\neq0,B=0.
\end{array}\right.$$
2) If $m_p=0$, then $${\Theta}_{b}(u,v)=\left\{\begin{array}{ll}
0, & {\rm if} \,\, b\in T,{\Tr}(\gamma)=0\,\,{\rm or}\,\,b\not\in T;\\
q^{\ell+\epsilon+1}, & {\rm if} \,\,b\in T,{\Tr}(\gamma^{q^e+1})=0,{\Tr}({\gamma})\neq0\,\,\\&{\rm or}\,\,{\Tr}(\gamma^{q^e+1})\neq0,{\Tr}({\gamma})\neq0,C\neq0;\\
-(q-1)q^{\ell+\epsilon+1}, & {\rm if} \,\,b\in T,{\Tr}(\gamma^{q^e+1})\neq0,{\Tr}({\gamma})\neq0,C=0,
\end{array}\right.$$
where $A={\Tr}(\gamma)^2-m_p{\Tr}(\gamma^{q^e+1})$, $B=\frac{u}{v^2-u m_p}{\Tr}(\gamma)^2+{\Tr}(\gamma^{q^e+1})$ and $C=\frac{u}{v^2}{\Tr}(\gamma)^2+{\Tr}(\gamma^{q^e+1})$.
\end{lem}
\begin{lem} \label{n11}
Let $b\in\mathbb{F}_{q^m}^*$, $u,v\neq0$, $v^2-um_p=0$, $m/\gcd(m,e)$ be even and $\epsilon$ be given as in \eqref{e1}. Recall the notation $T$ and $\gamma$ as in the beginning of Section \ref{s2}. Then the value of ${\Theta}_{b}(u,v)$ in \eqref{nn} is given by
$${\Theta}_{b}(u,v)=\left\{\begin{array}{ll}
0, & {\rm if}\,\,b\in T,{\Tr}(\gamma)^2-m_p{\Tr}(\gamma^{q^e+1})=0\,{\rm or}\, b\not\in T;\\
q^{\ell+\epsilon+1}, & {\rm if} \,\,b\in T,{\Tr}(\gamma^{q^e+1})=0,{\Tr}({\gamma})\neq0;\\
-(q-1)q^{\ell+\epsilon+1}{\eta}_{1}(-m_p{\Tr}(\gamma^{q^e+1})), & {\rm if} \,\, b\in T,{\Tr}(\gamma^{q^e+1})\neq0,{\Tr}({\gamma})=0;\\
q^{\ell+\epsilon+1}{\eta}_{1}({\Tr}(\gamma)^2-m_p{\Tr}(\gamma^{q^e+1})), & {\rm if} \,\, b\in T,{\Tr}(\gamma^{q^e+1})\neq0,{\Tr}({\gamma})\neq0,\\&\,\,\,\,{\Tr}(\gamma)^2-m_p{\Tr}(\gamma^{q^e+1})\neq0.
\end{array}\right.$$
\end{lem}

\begin{lem}\label{n12}
Let ${a_1}, {a_2} \in\fq^*$, $m/\gcd(m,e)$ be even and $\epsilon$ be given as in \eqref{e1}. Then
$$\begin{aligned}
|\{x\in \fqm:{\Tr}(x^{q^{e}+1})+\frac{{a_1}}{{a_2}}{{\Tr}(x)}^2=0\}|=q^{m-1}-(q-1)q^{\ell+\epsilon-1}{\eta}_{1}(1+\frac{{a_1} m_p}{{a_2}}).\nonumber\\
\end{aligned}$$
\end{lem}\begin{proof}
By the orthogonal relation of additive character, for any $x\in\fqm$, we have
\begin{align*}
&\frac{1}{q^2}\sum_{\pi\in\fq}(\sum_{w\in\fq}\chi_1(w(\pi-\Tr(x))))(\sum_{y\in\fq}\chi_1(y(\Tr(x^{q^{e}+1})+\frac{a_1\pi^2}{a_2}))))\\
=&\left\{\begin{array}{ll}
1, & \mbox{ if ${\Tr}(x^{q^{e}+1})+\frac{{a_1}}{{a_2}}{{\Tr}(x)}^2=0$};\\
0, & \mbox{ otherwise},
\end{array}\right.
\end{align*}
which is similar to the proof of \cite[Lemma 16]{TFF}.
Therefore, we have
\begin{align}\label{equa1}
N_1:=&\frac{1}{q^2}\sum_{x\in\fqm}\sum_{\pi\in\fq}(\sum_{w\in\fq}\chi_1(w(\pi-\Tr(x))))(\sum_{y\in\fq}\chi_1(y(\Tr(x^{q^{e}+1})+\frac{a_1\pi^2}{a_2}))))\nonumber\\
=&\frac{1}{q^2}\sum_{x\in\fqm}\sum_{\pi,w\in\fq}\chi_1(w(\pi-\Tr(x)))\nonumber\\
&+\frac{1}{q^2}\sum_{y\in\fq^*}\sum_{\pi,w\in\fq}\sum_{x\in\fqm}\chi_1(y\Tr(x^{q^{e}+1})+\frac{a_1y\pi^2}{a_2}+w(\pi-\Tr(x)))\nonumber\\
=&q^{m-1}+\Omega,
\end{align}
where $$\Omega:=\frac{1}{q^2}\sum_{y\in\fq^*}\sum_{\pi,w\in\fq}
\chi_1(\frac{a_1y\pi^2}{a_2}+w\pi)\sum_{x\in\fqm}
\chi(yx^{q^{e}+1}-wx).$$

By Lemmas \ref{m6}, \ref{m7} and \ref{n1}, we have
\[
\Omega=-q^{\ell+\epsilon-2}\sum_{y\in\fq^*}\sum_{w\in\fq}\sum_{\pi\in\fq}
\chi_1(\frac{a_1y\pi^2}{a_2}+w\pi-\frac{m_pw^2}{4y}).
\]
This together with Lemma \ref{m2} and the fact ${G}^2_{1}({\eta}_{1})={\eta}_{1}(-1)q$ gives
\begin{align}\label{equa2}
\Omega=&-q^{\ell+\epsilon-2}\sum_{y\in\fq^*}
\sum_{w\in\fq}G_1(\eta_1)\eta_1(\frac{a_1y}{a_2})
\chi_1(\frac{-a_2w^2}{4a_1y}-\frac{m_pw^2}{4y})\nonumber\\
=&-q^{\ell+\epsilon-2}G_1^2(\eta_1)\sum_{y\in\fq^*}\eta_1(\frac{a_1y}{a_2})\eta_1(\frac{-a_2}{4a_1y}-\frac{m_p}{4y})\nonumber\\
=&-(q-1)q^{\ell+\epsilon-1}{\eta}_{1}(1+\frac{{a_1} m_p}{{a_2}}).
\end{align}

By \eqref{equa1} and \eqref{equa2}, it completes the proof.
\end{proof}

\subsection{The proofs of Theorems \ref{t1} and \ref{t7}}
Here we only give the proofs of 1) and 2) of Theorem \ref{t1}, where $(u,v)=(0,0)$, and the remaining parts of Theorem \ref{t1} and Theorem \ref{t7} can be similarly proved.

It's obvious that $\C_{D_1}$ has length $n_{D_1}=N-1$, where $N$ is given as in \eqref{n2}. Then by Lemma \ref{lem1}, we have
\begin{align}\label{D1}
n_{D_1}&=\left\{\begin{array}{ll}
q^{m-2}-1, & {\rm if} \,\, m_p\neq0;\\
q^{m-2}-(q-1)q^{\ell+\epsilon-1}-1, & {\rm if} \,\,  m_p=0.
\end{array}\right.
\end{align}
Next we compute the Hamming weight $w_{H}(c_{b})$ of the codewords $c_{b}$ in $\C_{D_1}$. It is clear that $$w_H({c}_{b})=N-N_2,$$
where $$N_2:=|\{x\in \fqm: \Tr(x^{q^e+1})=u,\Tr(x)=v,\Tr(bx)=0 \}|.$$ Similar to the computation of $N$, it gives
\begin{align*}
N_2=&\frac{1}{q^3}\sum_{x\in\mathbb{F}_{q^m} }\sum_{y_1\in\mathbb{F}_{q}}\chi_1(y_1(\Tr(x^{q^{e}+1})-u))\sum_{y_2\in\mathbb{F}_{q}}
\chi_1(y_2(\Tr(x)-v))\sum_{y_3\in\mathbb{F}_{q}}
\chi_1(y_3\Tr(bx))
\nonumber\\
=&\frac{1}{q^3}\sum_{x\in\fqm}\sum_{{y}_{1}\in\mathbb{F}_{q}}\sum_{{y}_{2}\in
\mathbb{F}_{q}}\sum_{{y}_{3}\in
\mathbb{F}_{q}}\chi({y}_{1}x^{q^e+1}+({y}_{2}+{y}_{3}b)x)\mathbb{\chi}_{1}(-u{y}_{1}-v{y}_{2})\nonumber\\
=&\frac{1}{q^3}(\Theta_b(u,v)+\Omega_b(v)),
\end{align*}
where $\Theta_b(u,v)$ is defined as in \eqref{nn} and $${\Omega}_{b}(v):=\sum_{x\in\mathbb{F}_{q^m}}\sum_{{y}_{2}\in\mathbb{F}_{q}}\sum_{{y}_{3}
\in\mathbb{F}_{q}}\chi(({y}_{2}+{y}_{3}b)x)\mathbb{\chi}_{1}(-v{y}_{2}).$$

Assume that $(u,v)=(0,0)$ in the following proof. Recall that the codeword in ${\C}_{{D}_{1}}$ is given by ${c}_{b}={({\Tr}(bx))}_{x\in {D}_{1}}$ where $b\in \fqm$.
It is obvious that ${c}_{b}$ is a zero codeword if $b\in \fq$ since $\Tr(x)=0$ for $x\in D_1$. 
For $b\in \fqm\setminus\fq$, the Hamming weight $w_H({c}_{b})$ of $c_b$ in $\C_{D_1}$ is given by
\begin{eqnarray}\label{eq1}
w_H({c}_{b})=N-N_2=N-\frac{1}{q^3}(\Theta_b(u,v)+{\Omega}_{b}(v)),
\end{eqnarray}
where $N$ is determined by Lemma \ref{lem1}.
For $b\in \fqm\setminus\fq$, it gives ${\Omega}_{b}(0)=q^m$. For 1) and 2) of Theorem \ref{t1}, we study the weight distribution of ${\C}_{{D}_{1}}$ by considering the following two cases.

1) $m_p\neq0$. By \eqref{eq1} and Lemmas \ref{lem1} and \ref{n3}, the Hamming weight $w_H({c}_{b})$ is equal to
$$\left\{\begin{array}{ll}
0, & {\rm if} \,\,b\in\fq;\\
(q-1)q^{m-3}, & {\rm if}\,\,b\in T,b\notin\fq,\,A=0\,\,{\rm or}\,\,b\not\in T;\\
(q-1)(q^{m-3}+q^{\ell+\epsilon-2}), & {\rm if}\,\,b\in T,b\notin\fq,{\Tr}(\gamma^{q^u+1})=0,{\Tr}({\gamma})\neq0;\\
(q-1)(q^{m-3}+q^{\ell+\epsilon-2}\eta_1(-m_p{\Tr}(\gamma^{q^u+1}))), & {\rm if}\,\,b\in T,b\notin\fq,{\Tr}(\gamma^{q^u+1})\neq0,{\Tr}({\gamma})=0;\\
(q-1)(q^{m-3}+q^{\ell+\epsilon-2}\eta_1(A)), & {\rm if}\,\, b\in T,b\notin\fq,{\Tr}(\gamma^{q^u+1})\neq0,{\Tr}({\gamma})\neq0,A\neq0,
\end{array}\right.$$
where $A={\Tr}(\gamma)^2-m_p{\Tr}(\gamma^{q^e+1})$.
Let  ${w}_{1}:=(q-1)q^{m-3}$, ${w}_{2}:=(q-1)(q^{m-3}+q^{\ell+\epsilon-2})$ and ${w}_{3}:=(q-1)(q^{m-3}-q^{\ell+\epsilon-2})$.
Note that ${w}_{2}>{w}_{1}>{w}_{3}$, and $w_3>0$ due to $m>2\epsilon+2$. This shows that the minimum distance $d$ of $\C_{D_1}$ is equal to $(q-1)(q^{m-3}-q^{\ell+\epsilon-2})$ and the dimension of $\C_{D_1}$ is equal to $m-1$. By Lemmas \ref{mm1} and \ref{n12}, we have
$${A}_{{w}_{1}}=q^{m-1}-q^{m-2\epsilon-1}+q^{m-2\epsilon-2}-1.$$
Since $0\notin D_1$, the minimum distance of $\C^\perp_{D_1}$ is bigger than one, i.e., $A^\perp_1=0$. Then by Pless power moments (see \cite{HP}, p.256), we have $$\left\{\begin{array}{ll}
{A}_{{w}_{1}}+{A}_{{w}_{2}}+{A}_{{w}_{3}}=q^{m-1}-1,\\
{w}_{1}{A}_{{w}_{1}}+{w}_{2}{A}_{{w}_{2}}+{w}_{3}{A}_{{w}_{3}}=q^{m-2}(q-1)n_{D_1}.
\end{array}\right.$$
Then it gives that ${A}_{{w}_{2}}=\frac{(q-1)}{2}(q^{m-2\epsilon-2}-q^{\ell-\epsilon-1})$ and ${A}_{{w}_{3}}=\frac{(q-1)}{2}(q^{m-2\epsilon-2}+q^{\ell-\epsilon-1})$.

2) $m_p=0$. By \eqref{eq1} and Lemmas \ref{lem1} and \ref{n3}, we have  $$w_H({c}_{b})=\left\{\begin{array}{ll}
0, & {\rm if} \,\, b\in\fq;\\
{w}_{1}:=(q-1)q^{m-3}, & {\rm if}\,\,b\in T,b\notin\fq,{\Tr}(\gamma)=0,{\Tr}(\gamma^{q^e+1})=0;\\
{w}_{2}:=(q-1)(q^{m-3}-q^{\ell+\epsilon-1}), & {\rm if}\,\,b\in T,b\notin \fq,{\Tr}(\gamma)=0, {\Tr}(\gamma^{q^e+1})\neq0;\\
{w}_{3}:=(q-1)(q^{m-3}+q^{\ell+\epsilon-2}-q^{\ell+\epsilon-1}), &{\rm if}\,\,b\in T,b\notin\fq,{\Tr}(\gamma)\neq0 \,\,{\rm or}\,\,b\not\in T.
\end{array}\right.$$
Note that ${w}_{1}>{w}_{3}>{w}_{2}$, and $w_2>0$ due to $m>3$. This shows that the minimum distance $d$ of $\C_{D_1}$ is $(q-1)(q^{m-3}-q^{\ell+\epsilon-1})$ and the dimension of $\C_{D_1}$ is $m-1$. By Lemmas \ref{mm1} and \ref{lem1}, we can get that ${A}_{{w}_{1}}=q^{m-2\epsilon-3}-(q-1)q^{\ell-\epsilon-2}-1$ and ${A}_{{w}_{3}}=q^{m-1}-q^{m-2\epsilon-2}$.
From Pless power moments (see \cite{HP}, p.256), it leads to ${A}_{{w}_{2}}=(q-1)(q^{m-2\epsilon-3}+q^{\ell-\epsilon-2})$.
This completes the proof.

\subsection{The proofs of Theorems \ref{t4}, \ref{t10}, \ref{t5} and \ref{t11}}
In this subsection, we investigate the linear codes  $\C_{D_2}$ and $\C_{D_3}$ of the form \eqref{000}  with defining sets ${D_2}$ and ${D_3}$. Here we only give the proof of Theorem \ref{t4}, where $(u,v)=(0,0)$, and Theorems \ref{t10}, \ref{t5} and \ref{t11} can be similarly proved. Let ${E}_{1}=\{x\in \mathbb{F}^*_{q^m}:{\Tr}(x^{q^e+1})=u\}$ and ${E}_{2}=\{x\in \mathbb{F}^*_{q^m}:{\Tr}(x)=v\}$, which implies that ${D}_{2}={E}_{1}\setminus{D}_{1}$ and ${D}_{3}={D}_{2}\cup{E}_{2}$. Assume that $(u,v)=(0,0)$ in the following proof.

By the definition of $E_1$, the length of $\C_{E_1}$ is given by ${n}_{{E}_{1}}=|\{x\in \mathbb{F}^*_{q^m}:{\Tr}(x^{q^e+1})=0\}|$. By Lemmas \ref{m4} and \ref{n1}, we have
\begin{align}\label{D2}
{n}_{{E}_{1}}=&\frac{1}{q}\sum_{x\in\fqm }\sum_{y\in\mathbb{F}_{q}}\chi(yx^{q^e+1})-1\nonumber\\=&q^{m-1}+\frac{1}{q}\sum_{x\in\fqm }\sum_{y\in\mathbb{F}^*_{q}}\chi(yx^{q^e+1})-1\nonumber\\=&q^{m-1}-(q-1)q^{\ell+\epsilon-1}-1.
\end{align}
By \eqref{D1} and \eqref{D2}, the length of ${\C}_{{D}_{2}}$ is
$${n}_{{D}_{2}}={n}_{{E}_{1}}-{n}_{{D}_{1}}=\left\{\begin{array}{ll}
(q-1)(q^{m-2}-q^{\ell+\epsilon-1}), & {\rm if} \,\, m_p\neq0;\\
q^{m-1}-q^{m-2}, & {\rm if} \,\, m_p=0.\\
\end{array}\right.$$
For $b\in\fqm^*$, define $\Psi_b=|\{x\in\fqm:{\Tr}(x^{q^e+1})=0,{\Tr}(bx)=0\}|$. Similar to the proof of Lemma \ref{n3}, we have
$$\Psi_b=\left\{\begin{array}{ll}
q^{m-2}-(q-1)q^{\ell+\epsilon-1}, & {\rm if} \,\, b\in T,{\Tr}(\gamma^{q^e+1})=0;\\
q^{m-2}, & {\rm if} \,\, b\in T,{\Tr}(\gamma^{q^e+1})\neq0;\\
q^{m-2}-(q-1)q^{\ell+\epsilon-2}, & {\rm if} \,\, b\notin T.\\
\end{array}\right.$$
Let ${w}_{H}({\hat{c}}_{b})$ denote the Hamming weight of the codeword ${\hat{c}}_{b}$ in $\C_{E_1}$. It is clear that
 $${w}_{H}({\hat{c}}_{b})={n}_{{E}_{1}}+1-\Psi_b=\left\{\begin{array}{ll}
(q-1)q^{m-2}, & {\rm if} \,\, b\in T,{\Tr}(\gamma^{q^e+1})=0;\\
(q-1)(q^{m-2}-q^{\ell+\epsilon-1}), & {\rm if} \,\, b\in T,{\Tr}(\gamma^{q^e+1})\neq0;\\
(q-1)q^{m-2}-(q-1)^2q^{\ell+\epsilon-2}, & {\rm if} \,\, b\notin T.\\
\end{array}\right.$$

Let ${w}_{H}({\bar{c}}_{b})$ denote the Hamming weight of the codeword ${\bar{c}}_{b}$ in $\C_{D_2}$. Due to ${D}_{2}={E}_{1}\setminus{D}_{1}$, the Hamming weights of ${c}_{b},{\bar{c}}_{b}$ and  ${\hat{c}}_{b}$ satisfy that
 $${w}_{H}({\bar{c}}_{b})={w}_{H}({\hat{c}}_{b})-{w}_{H}({c}_{b}).$$
Next we compute the weight distribution of $\C_{D_2}$ by considering the following two cases.

1) $m_p\neq0$. The Hamming weight $w_H({\bar{c}}_{b})$ is equal to
$$\left\{\begin{array}{ll}
0, & {\rm if} \,\, b=0;\\
{w}_{1}:=(q-1)(q^{m-2}-q^{\ell+\epsilon-1}), & {\rm if} \,\, b\in\fq^*;\\
{w}_{2}:=(q-1)^2q^{m-3}, &{\rm if} \,\,b\in T,b\notin\fq,{\Tr}(\gamma^{q^e+1})=0,{\Tr}(\gamma)=0;\\
{w}_{3}:=(q-1)(q^{m-2}-q^{m-3}-q^{\ell+\epsilon-2}), & {\rm if}\,\,b\in T,b\notin\fq,{\Tr}(\gamma^{q^e+1})=0,{\Tr}(\gamma)\neq0;\\
{w}_{4}:=(q-1)(q^{m-2}-q^{m-3}-q^{\ell+\epsilon-1}), & {\rm if}\,\,b\in T,b\notin\fq,{\Tr}(\gamma^{q^e+1})\neq0,{\Tr}(\gamma)\neq0,\\&\,\,\,\,{\Tr}(\gamma)^2-m_p{\Tr}(\gamma^{q^e+1})=0 ;\\
{w}_{5}:=(q-1)(q^{m-2}-q^{\ell+\epsilon-1}-q^{m-3}-q^{\ell+\epsilon-2}),& {\rm if} \,\,b\in T,b\notin\fq,{\Tr}(\gamma^{q^e+1})\neq0,{\Tr}(\gamma)=0\\&\,\,\,\,\,\,{\eta}_{1}(-m_p{\Tr}(\gamma^{q^e+1}))=1\,\,{\rm or}\\&\,\,\,\,{\Tr}(\gamma^{q^e+1})\neq0,{\Tr}(\gamma)\neq0,{\eta}_{1}(A)=1;\\
{w}_{6}:=(q-1)^2(q^{m-3}-q^{\ell+\epsilon-2}), & \mbox{otherwise},
\end{array}\right.$$
where $A={\Tr}(\gamma)^2-m_p{\Tr}(\gamma^{q^e+1})$. Note that $w_1>w_2>w_3>w_6>w_4>w_5$, and $w_5>0$ due to $m>2\epsilon+2$. This shows that the minimum distance $d$ of $\C_{D_2}$ is $(q-1)(q^{m-2}-q^{\ell+\epsilon-1}-q^{m-3}-q^{\ell+\epsilon-2})$ and the dimension of ${\C}_{{D}_{2}}$ is $m$.

By Lemmas \ref{lem1} and \ref{n12}, it's clear that ${A}_{{w}_{1}}=q-1$, ${A}_{{w}_{2}}=q^{m-2\epsilon-2}-1$, ${A}_{{w}_{3}}=(q-1)(q^{m-2\epsilon-2}-q^{\ell-\epsilon-1})$ and
${A}_{{w}_{4}}=|\{b\in\fqm:b\in T,{\Tr}(\gamma)^2-m_p{\Tr}(\gamma^{q^e+1})=0\}|-{A}_{{w}_{2}}-|\{b\in\fq\}|=(q-1)(q^{m-2\epsilon-2}-1)$.
From Pless power moments (see \cite{HP}, p.256), we can get that ${A}_{{w}_{5}}=\frac{q-1}{2}(2q^{\ell-\epsilon-1}-2q^{m-2\epsilon-2}+q^{m-2\epsilon-1}-q^{\ell-\epsilon})$ and ${A}_{{w}_{6}}=q^m+\frac{1}{2}(q^{\ell-\epsilon+1}-q^{m-2\epsilon}-q^{m-2\epsilon-1}-q^{\ell-\epsilon})$.

2) $m_p=0$. The Hamming weight $w_H({\bar{c}}_{b})$ is equal to
$$w_H({\bar{c}}_{b})=\left\{\begin{array}{ll}
0, & {\rm if} \, b=0;\\
{w}_{1}:=(q-1)q^{m-2}, & {\rm if} \,\, b\in\mathbb{F}^*_{q};\\
{w}_{2}:=(q-1)^2(q^{m-3}+q^{\ell+\epsilon-2}), & {\rm if} \,\,b\in T,b\notin\fq,{\Tr}(\gamma)\neq0,{\Tr}(\gamma^{q^e+1})=0 ;\\
{w}_{3}:=(q-1)(q^{m-2}-q^{m-3}-q^{\ell+\epsilon-2}), & {\rm if}\,\,b\in T,b\notin\fq,{\Tr}(\gamma)\neq0, {\Tr}(\gamma^{q^e+1})\neq0;\\
{w}_{4}:=(q-1)^2q^{m-3}, & {\rm if}\,\,b\in T,b\not\in \fq,{\Tr}(\gamma)=0 \,\,{\rm or}\,\,b\not\in T.
\end{array}\right.$$
Note that $w_1>w_2>w_4>w_3$, and $w_3>0$ due to $m>3$. This shows that the minimum distance $d$ of $\C_{D_2}$ is equal to $(q-1)(q^{m-2}-q^{m-3}-q^{\ell+\epsilon-2})$ and the dimension of ${\C}_{{D}_{2}}$ is $m$.
It is clear that ${A}_{{w}_{1}}=q-1$, ${A}_{{w}_{2}}=(q-1)q^{m-2\epsilon-2}$, and ${A}_{{w}_{4}}=q^m-q^{m-2\epsilon}+q^{m-2\epsilon-1}-q$.
Then ${A}_{{w}_{3}}=(q-1)^2q^{m-2\epsilon-2}$ from Pless power moments (see \cite{HP}, p.256). This completes the proof.

\begin{table}\footnotesize
\caption{Some good $[n,k,d]$ linear codes over $\fq$ obtained in this paper}\label{tabb}
\newcommand{\tabincell}[2]{\begin{tabular}{@{}#1@{}}#2\end{tabular}}
\centering
\begin{tabular}{c c c c c}
\hline
$(q,m,e)$   & $(u,v)$ & Parameters  & Optimal?  &This paper                                \\ \hline
$(3,4,1)$       &$(0,0)$    &$(8,2,6)$    &yes  &Theorem \ref{t1}\\
$(5,4,1)$      &$(0,0)$    &$(24,2,20)$  &yes  &Theorem \ref{t1}     \\
$(7,4,1)$    &$(0,0)$    &$(48,2,42)$  &yes  &Theorem \ref{t1}     \\
$(9,4,1)$    &$(0,0)$    &$(80,2,72)$  &yes   &Theorem \ref{t1}    \\
$(3,4,2)$    &$(0,0)$    &$(8,3,4)$  &almost optimal &Theorem \ref{t1}    \\
$(5,4,2)$    &$(1,0)$    &$(20,3,14)$  &almost optimal   &Theorem \ref{t1}    \\
$(9,4,2)$    &$(1,0)$    &$(72,3,62)$  &almost optimal   &Theorem \ref{t1}  \\
$(3,4,1)$    &$(1,0)$    &$(18,3,12)$  &yes     &Theorem \ref{t1}  \\
$(5,4,1)$    &$(2,0)$    &$(50,3,40)$  &yes     &Theorem \ref{t1}  \\
$(7,4,1)$    &$(1,0)$    &$(98,3,84)$  &yes     &Theorem \ref{t1}  \\
$(3,4,2)$    &$(0,2)$    &$(6,4,2)$  &yes     &Theorem \ref{t7}  \\
$(5,4,2)$    &$(0,2)$    &$(20,4,14)$  &yes     &Theorem \ref{t7}  \\
$(7,4,2)$    &$(0,2)$    &$(42,4,34)$  &yes     &Theorem \ref{t7}  \\
$(9,4,2)$    &$(0,2)$    &$(72,4,62)$  &yes     &Theorem \ref{t7}  \\
$(3,6,1)$    &$(0,2)$    &$(81,6,51)$  &yes     &Theorem \ref{t7}  \\
$(3,4,2)$    &$(2,2)$    &$(12,4,6)$  &yes     &Theorem \ref{t7}  \\
$(3,4,2)$    &$(1,1)$    &$(9,4,4)$  &almost optimal     &Theorem \ref{t7}  \\
$(3,4,1)$    &$(1,1)$    &$(9,3,6)$  &yes     &Theorem \ref{t7}  \\
$(9,4,1)$    &$(1,1)$    &$(81,3,72)$  &yes     &Theorem \ref{t7}  \\ \hline
\end{tabular}
\end{table}

\section{Concluding remarks} \label{s5}
In this paper, we investigated the $q$-ary linear codes defined by \eqref{000} and \eqref{001}. With detailed computation, we obtained several classes of $t$-weight linear codes over $\fq$ with flexible parameters, where $t=3,4,5,6$. The parameters and weight distributions of these codes were completely determined by using Weil sums and Gauss sums. Moreover, from our constructions, several classes of optimal linear codes meeting the Griesmer bound were derived (see Corollaries \ref{b1} and \ref{b3}), and some (almost) optimal codes can be produced as shown in Table \ref{tabb}.

\section*{Acknowledgements}
This work was supported by the Major Program(JD) of Hubei Province (No. 2023BAA027), the National Natural Science Foundation of China (No. 12401688), the Natural Science Foundation of Hubei Province of China (No. 2024AFB419) and the innovation group project of the natural science foundation of Hubei Province of China (No. 2023AFA021).

\end{document}